\documentclass[journal]{IEEEtran}
\IEEEoverridecommandlockouts
\usepackage{amsthm}

\usepackage{color}
\usepackage{amsmath}
\usepackage{graphicx, subfigure}
\usepackage{amssymb}
\usepackage{cases}
\usepackage{cite}
\usepackage{url}
\usepackage{algorithm}
\usepackage{algorithmic}
\usepackage{relsize}
\usepackage[nodisplayskipstretch]{setspace}
\usepackage{longtable}
\usepackage{booktabs}
\usepackage{multirow}
\usepackage{bm}
\usepackage{stfloats}
\newtheorem{theorem}{Theorem}
\newtheorem{lemma}{Lemma}
\newtheorem{remark}{Remark}
\newtheorem{definition}{Definition}
\newtheorem{corollary}{Corollary}
\newtheorem{proposition}{Proposition}

\theoremstyle{plain}

\begin{document}

\title{Scalable Uplink Signal Detection in C-RANs via Randomized Gaussian Message Passing}
\author{Congmin Fan, Xiaojun Yuan, Ying Jun (Angela) Zhang
\IEEEcompsocitemizethanks{
\IEEEcompsocthanksitem
This work has been presented in part at the 2015 IEEE GLOBECOM, San Diego, USA, and the 2016 IEEE ICC, Kuala Lumpur, Malaysia. This work was supported in part by General Research Funding (Project number 14209414) from the Research Grants Council of Hong Kong and by the National Basic Research Program (973 program Program number 2013CB336701). The work of X. Yuan was supported in part by the China Recruitment Program of Global Young Experts.
\IEEEcompsocthanksitem
C. Fan are with the Department of Information Engineering, The Chinese University of Hong Kong.
\IEEEcompsocthanksitem
X. Yuan was with the School of Information Science and Technology, ShanghaiTech University, Shanghai, China. He is now with the National Key Laboratory of Science and Technology on Communications, the University of Electronic Science and Technology of China, Chengdu, China.
\IEEEcompsocthanksitem
Y. J. Zhang is with the Department of Information Engineering, The Chinese University of Hong Kong. She is also with the Institute of Network Coding (Shenzhen), Shenzhen Research Institute, The Chinese University of Hong Kong.}
}

\maketitle

\begin{abstract}
Cloud Radio Access Network (C-RAN) is a promising architecture for unprecedented capacity enhancement in next-generation wireless networks thanks to the centralization and virtualization of base station processing. However, centralized signal processing in C-RANs involves high computational complexity that quickly becomes unaffordable when the network grows to a huge size. Among the first, this paper endeavours to design a \textit{scalable} uplink signal detection algorithm, in the sense that both the complexity per unit network area and the total computation time remain constant when the network size grows. To this end, we formulate the signal detection in C-RAN as an inference problem over a bipartite random geometric graph. By passing messages among neighboring nodes, message passing (a.k.a. belief propagation) provides an efficient way to solve the inference problem over a sparse graph. However, the traditional message-passing algorithm is not guaranteed to converge, because the corresponding bipartite random geometric graph is locally dense and contains many short loops. As a major contribution of this paper, we propose a randomized Gaussian message passing (RGMP) algorithm to improve the convergence. Instead of exchanging messages simultaneously or in a fixed order, we propose to exchange messages asynchronously in a random order. The proposed RGMP algorithm demonstrates significantly better convergence performance than conventional message passing. The randomness of the message updating schedule also simplifies the analysis, and allows the derivation of the convergence conditions for the RGMP algorithm. In addition, we generalize the RGMP algorithm to a blockwise RGMP (B-RGMP) algorithm, which allows parallel implementation. The average computation time of B-RGMP remains constant when the network size increases. 

\end{abstract}

\section{Introduction}
\subsection{Background and Motivations}
Cloud Radio Access Networks (C-RANs) have drawn considerable attention for their potential to sustain the explosive traffic demand in wireless communications. Unlike traditional cellular networks, a C-RAN splits the low-cost and light-weighted remote radio heads (RRHs) from the baseband processing units (BBUs), and merges the latter into a data center. The RRHs and BBUs are connected by a low-latency, high-bandwidth fiber network. The special architecture of C-RAN allows full-scale RRH coordination, which enables flexible interference management, dynamic resource allocation, and collaborative radio technology. This consequently leads to significant capacity enhancement. The full-scale coordination, however, also introduces a severe complexity issue. The state-of-the-art C-RAN technology is able to support thousands of RRHs \cite{mobile2011c}. Full-scale RRH coordination over such a large network involves prohibitively high computational complexity. For example, the linear minimum mean square error (MMSE) detector requires cubic complexity in the network size (in terms of the number of RRHs), or equivalently a quadratic complexity normalized by the number of RRHs\cite{tuchler2002minimum}. This implies that the detection complexity quickly becomes unaffordable as the network size grows. As such, a main challenge of C-RAN is to design \textit{scalable} coordination algorithms, where \textit{scalable} means: 1) the performance is near the optimum performance of full-scale RRH coordination, 2) the normalized computational complexity per RRH does not grow with the network size, or equivalently, the total computational complexity grows linearly with the network size, 3) with parallel implementation, the total computation time remains constant when the network size grows.
\par 
In a C-RAN, users and RRHs are scattered over a large area. Due to the propagation attenuation of electromagnetic waves, an RRH usually receives relatively strong signals from only a small number of nearby users. Moreover, the transmission delay prevents the signals from far-away users to be processed. Intuitively, ignoring the signals from far-away users in general does not cause much performance loss. As shown in our previous work \cite{fan2014dynamic}, with a distance-threshold-based channel sparsification approach, a vast majority of signals over the transmission links can be ignored if we can tolerate a small degradation in the signal-to-noise-plus-interference ratio (SINR). As such, each RRH only needs to serve its nearby users whose distances to the RRH are below a certain threshold. Based on the sparsified channel matrix, \cite{fan2014dynamic} proposes an algorithm that greatly reduces the computational complexity of MMSE detection from $O(N^3)$ to $O(N^a)$, where $N$ is the total number of RRHs and $a \in (1,2]$ is a constant determined by the computation implementations. Yet, the algorithm is still not perfectly scalable, in the sense that the complexity grows faster than linear with the network size. 
\par 
\begin{figure*}[!h]
\centering
\subfigure[C-RAN architecture]
{\includegraphics[width=0.54\textwidth]{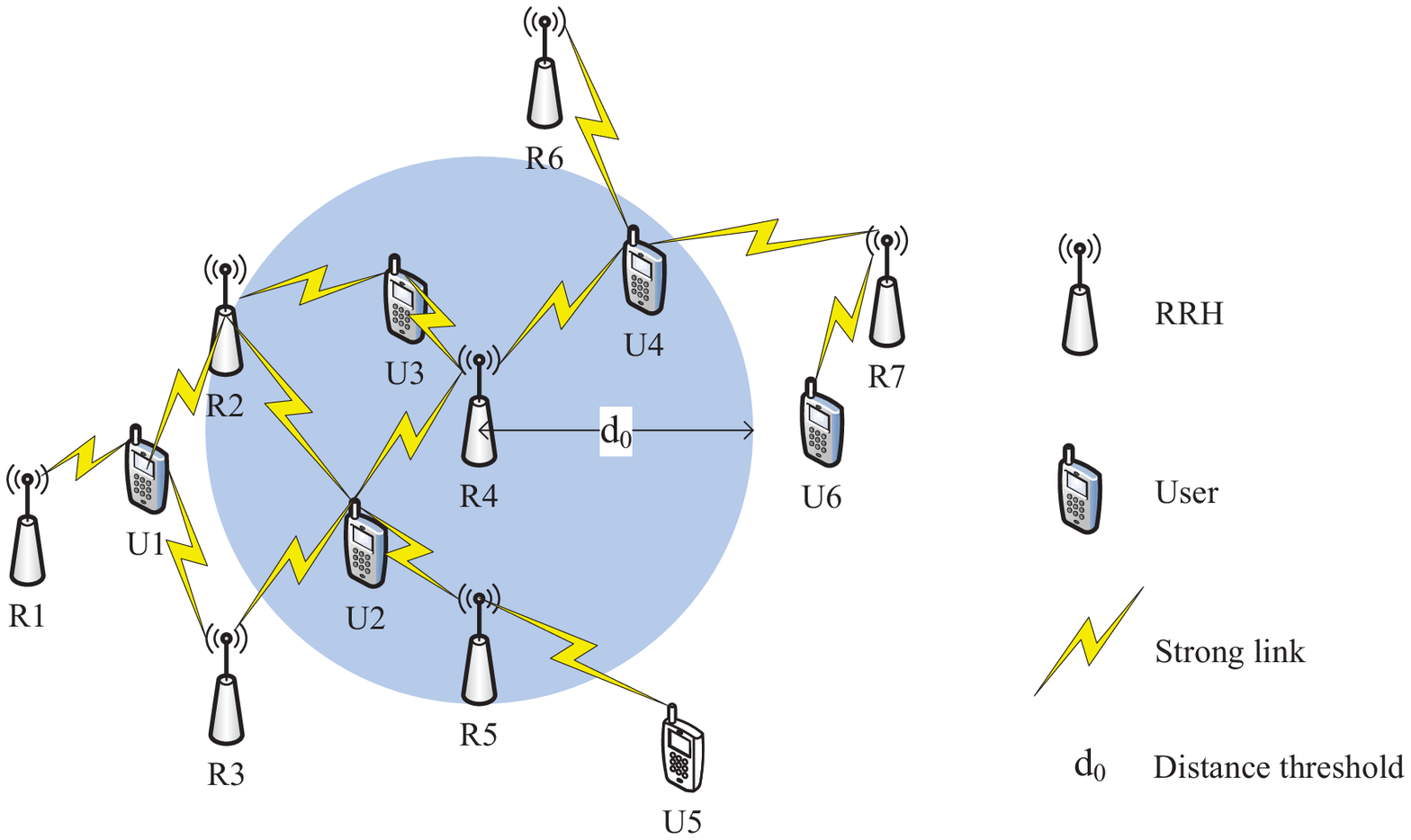}
\label{fig:sfg_1}}
\hspace{0.1in}
\subfigure[Bipartite random geometric graph]
{\includegraphics[width=0.38\textwidth]{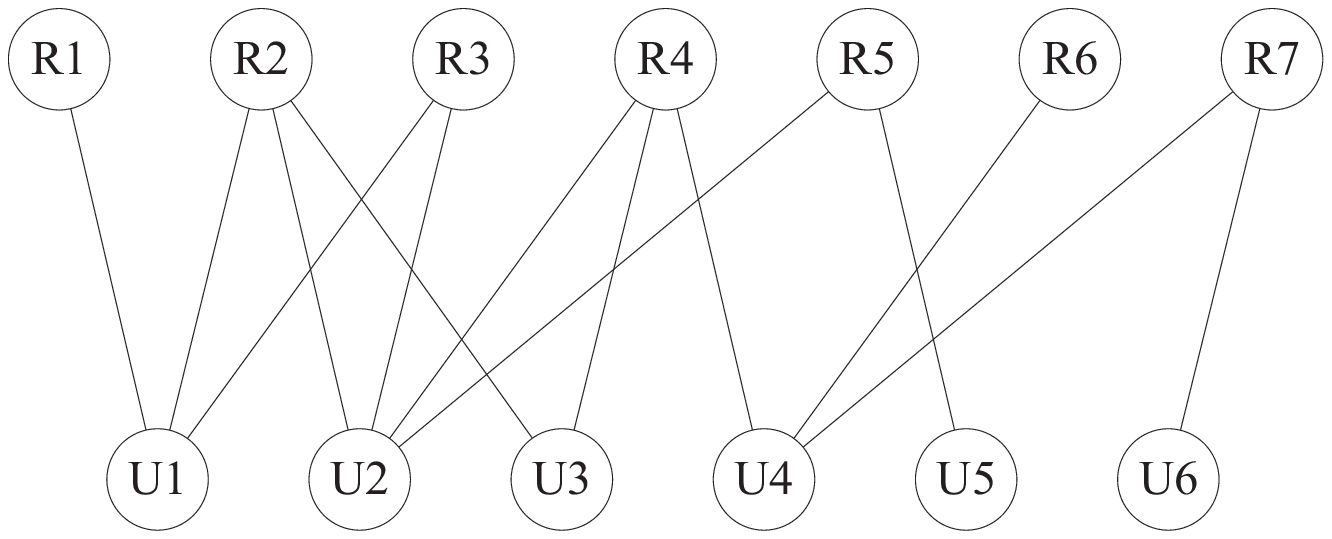}\label{fig:sfg_2}}
\caption{Graphical representation of a C-RAN}\label{fig:sfg}
\end{figure*}
In this paper, we are interested in designing a perfectly scalable algorithm for joint signal detection in the uplink of C-RAN. With channel sparsification \cite{fan2014dynamic}, a C-RAN system can be represented by a bipartite random geometric graph, as shown in Fig.~\ref{fig:sfg}. Here, RRHs and users are treated as vertices/nodes, and an edge connects an RRH and a user if the distance between them does not exceed the threshold. Then, signal detection in a C-RAN is converted to a statistical inference problem over a bipartite random geometric graph, where the inference problem is to estimate the signals from unobserved nodes (i.e., user nodes) conditional on signals from observed nodes (i.e., RRH nodes).

\par 
We propose to solve the above inference problem in C-RAN using message passing (a.k.a. belief propagation), an iterative algorithm well-known for its good performance and low complexity; see, e.g., \cite{du2013network, du2013distributed,kschischang2001factor} and the references therein. In a message-passing algorithm, messages are exchanged between nodes with edge connection. Thus, the complexity of message passing is proportional to the number of edges in the network. In a C-RAN with channel sparsification, the number of messages per RRH is proportional to the number of nearby users in its neighborhood, which does not scale with the network size. Thus, the total complexity per iteration of message passing in C-RANs is linear in the network size. Unfortunately, the convergence of message passing over a bipartite random geometric graph is not ensured. The reason is that a random geometric graph is locally dense and always contains loops. It is well-known that message passing is not guaranteed to converge when the graph is loopy \cite{weiss2001correctness}. Several sufficient convergence conditions for message passing have been derived, such as diagonal dominant\cite{weiss2001correctness}, walk-summability\cite{malioutov2006walk}, and convex decomposition \cite{moallemi2009convergence}. More recently, a necessary and sufficient condition has been derived by Su and Wu in \cite{su2015convergence}. However, message passing for C-RANs is not guaranteed to satisfy any of the above-mentioned conditions. Indeed, numerical simulations indicate a non-trivial probability that the message-passing algorithm for C-RANs does not converge. 
\subsection{Contributions}
In this paper, we propose a randomized message updating schedule for GMP, to address the convergence issue of message passing over a bipartite random geometric graph. The corresponding GMP algorithm is called randomized Gaussian message passing (RGMP). Unlike conventional message passing with synchronous message updating, the RGMP algorithm updates messages serially in a random order. To the best of our knowledge, this is the first work to introduce random serial updating for GMP over a bipartite random geometric graph. Intuitively, when messages are exchanged among nodes of a loopy graph, errors may accumulate along loops, which eventually leads to the divergence of the algorithm \cite{ihler2005loopy}. Updating messages sequentially in a random order weakens the effect of loops and thus improves the convergence. The randomness of the message update schedule also simplifies the analysis, and allows us to derive a necessary and sufficient condition for the expected convergence of the RGMP algorithm. We show by both analysis and numerical results that the RGMP algorithm converges with a much higher probability than conventional message passing. Indeed, we have never observed a single case of divergence in simulation when the network size is moderately large (i.e., when the network has more than five RRHs). Our numerical results also indicate that the number of iterations of RGMP does not increase with the network size. This implies that the total computational complexity is linear with the network size. Moreover, we extend the RGMP algorithm to a blockwise RGMP (B-RGMP) algorithm, which allows parallel implementation. It is proved that when the messages are updated within two blocks, the convergence condition of B-RGMP is less stringent than that of GMP. That is, when the GMP algorithm converges, the expected output of the B-RGMP algorithm always converges. We also observe that the average computation time of the B-RGMP algorithm remains constant when the network size increases. Therefore, B-RGMP is perfectly scalable in terms of computation time. To the best of our knowledge, our paper is the first work to achieve perfect scalability in terms of computation time of joint signal detection in C-RAN systems. 
\subsection{Related Work}\label{testtest}
A C-RAN is similar to a multiuser multiple-input multiple-output (MU-MIMO) system if the cooperative RRHs are regarded as multiple antennas of a single base station. Signal processing has been extensively studied in MU-MIMO systems. However, limited research has been focused on the scalability of signal processing complexity in MU-MIMO. Moreover, the distributed locations of RRHs make the distribution of the channel matrix in C-RAN distinct from that in MU-MIMO. As such, many existing results in MU-MIMO do not hold in C-RANs. For example, \cite{liu2015low,liu2016convergence,liu2016gaussian} proved the convergence of a message-passing-based detection algorithm for massive MU-MIMO system by exploiting the law of large numbers and the random matrix theory. However, in C-RANs, channel coefficients are dependent of each other since users and RRHs are geographically related to each other. Furthermore, even if we make the assumption that the channel coefficients are independent, the existing random matrix theory still does not apply, since the channel coefficients follow a truncated heavy-tailed distribution (which is not covered in the existing random matrix theory). A widely used method to improve the convergence of message passing is the damping technique \cite{liu2016convergence,liu2016gaussian,sohn2011belief,som2010improved,moretti2014on}. With damping, an updated message is a weighted average of the message in the previous round of iteration and the mesage calculated by the original message updating rules. The weight in fact controls the trade-off between the convergence speed and the convergence probability. However, how to efficiently determine the value of the weight is still an open problem. It is also well-known that the schedule of message updating affects the convergence property of GMP \cite{Goldberger2008serial}. Ref. \cite{Goldberger2008serial} analysed the average convergence speed of random serial update schedules for loop-free factor graphs. It has been proved that GMP with random serial schedules converges about twice as fast as the conventional GMP. The schedule analysed in \cite{Goldberger2008serial} is randomly chosen and fixed in each realization instead of for each iteration. That is, the update schedule is the same for all iterations in \cite{Goldberger2008serial}. As shown in our later simulations, the convergence of serial GMP heavily depends on the update order. With a randomly picked order, the serial GMP proposed in \cite{Goldberger2008serial} does not ensure convergence. Another variant of message passing is approximate message passing (AMP). AMP was first proposed as a low-complexity iterative algorithm for compressed sensing \cite{donoho2009message}. Then, Rangan extended AMP to a general algorithm, named generalized approximate message passing (GAMP) \cite{rangan2011generalized}. However, in this paper, we show by numerical simulations that for GAMP-based signal detection in C-RANs, the number of iterations needed for convergence is roughly linear in the network size. This translates to quadratic computational complexity in total, implying that GAMP is not scalable.
\par
In our previous work \cite{fan2014dynamic}, we proposed a dynamic clustering algorithm to reduce the computational complexity of the MMSE detector. The complexity of the algorithm is reduced from cubic to no more than quadratic in the number of RRHs. In \cite{shi2015large}, Shi \textit{et al.} presented a two-stage approach to solve large-scale convex optimization problems for dense wireless cooperative networks, such as C-RANs. Matrix stuffing and alternating direction method of multipliers (ADMM) were used to speed up the computation. In addition, it was shown in \cite{sun2015on} that the expected output of randomly permuted ADMM converges to the unique solution of the optimal linear detector. In this paper we show that the ADMM algorithm converges much more slowly than the proposed RGMP algorithm when applied to large networks like C-RANs. 
\subsection{Organization}
The rest of the paper is organized as follows. In Section II, we describe the system model. In Section III, we introduce a Gaussian message-passing algorithm with channel sparsification for signal detection in C-RANs with linear complexity per iteration, and then discuss the convergence issue. In Section IV, we propose the RGMP algorithm to address the convergence issue of Gaussian message passing. In Section V, the convergence condition of the RGMP algorithm is analysed. In Section VI, RGMP is extended to the B-RGMP algorithm, which can significantly reduce the computation time through parallel implementation. In Section VII, simulation results are demonstrated to compare RGMP and B-RGMP with other existing algorithms. Conclusions and future works are discussed in Section VIII.
\section{System Model}
In this paper we consider the uplink transmission of a C-RAN with $N$ single-antenna RRHs and $K$ single-antenna users. Suppose that both the RRHs and the users are randomly located over an area. Let $x_k$ be the signal transmitted by user $k$, and $y_n$ be the received signal at RRH $n$. Denote $\mathbf{x}=[x_1,\cdots, x_K]^T$ and $\mathbf{y}=[y_1,\cdots,y_N]^T$. Then, the received signal vector $\mathbf{y}\in \mathbb{C}^{N \times 1}$ at the RRHs is
\begin{equation}
\mathbf{y}=P^{\frac{1}{2}}\mathbf{H}\mathbf{x}+\mathbf{n},
\label{eqn:model}
\end{equation}
where $\mathbf{H} \in \mathbb{C}^{N \times K}$ denotes the channel matrix, with the $(n, k)$-th entry $H_{n,k}$ being the channel coefficient between the $k$-th user and the $n$-th RRH; $P$ is the transmission power allocated to each user; and $\mathbf{n} \sim \mathcal{CN}(\mathbf{0},N_0\mathbf{I})$ is a noise vector received by the RRHs. The transmitted signals are assumed to have zero mean and unit variance, i.e., $E[\mathbf{x}]=\mathbf{0}$ and $E[\mathbf{x}\mathbf{x}^H]=\mathbf{I}$. We further assume $H_{n,k}=\gamma_{n,k}d_{n,k}^{-\frac{\alpha}{2}}$, where $\gamma_{n,k}$ is the i.i.d. Rayleigh fading coefficient with zero mean and unit variance, $d_{n,k}$ is the distance between the $k$-th user and the $n$-th RRH, and $\alpha$ is the path loss exponent. Here, $d_{n,k}^{-\alpha}$ is the path loss from the $k$-th user to the $n$-th RRH.
\par 
In this paper, we employ linear MMSE detection to estimate the transmitted signal vector $\mathbf{x}$, with the decision statistics given by
\begin{equation}
\widehat{\mathbf{x}}=P^{\frac{1}{2}}\mathbf{H}^H(P\mathbf{H}\mathbf{H}^H +  N_0\mathbf I)^{-1}\mathbf{y}.
\label{eqn:x}
\end{equation}
In the above, the inversion of the $N\times N$ matrix $P\mathbf{H}\mathbf{H}^H +  N_0\mathbf I$ requires computational complexity of $O(N^3)$. This complexity is prohibitively high for a large-scale C-RAN with hundreds and thousands of RRHs, thus posing a serious scalability problem. In what follows, we endeavour to develop a scalable algorithm to estimate $\mathbf{x}$ by MMSE detection with complexity $O(N)$ under the assumption that $K$ grows at the same rate as $N$ (i.e., the ratio between $N$ and $K$ is fixed). In other words, the average computational complexity per RRH (or per unit network size) does not scale with $N$. 
\section{Gaussian Message Passing with Channel Sparsification}
In this section, we first describe the channel sparsification approach introduced by the authors in \cite{fan2014dynamic} to model a C-RAN as a bipartite random geometric graph. Then, we apply the Gaussian message-passing algorithm proposed in \cite{bickson2008gaussian} over bipartite random geometric graphs for signal detection.
\subsection{Channel Sparsification}
We borrow the channel sparsification approach in our recent work \cite{fan2014dynamic} to sparsify the channel matrix, as described below. The entries of $\mathbf{H}$ are discarded based on the distances between RRHs and users. Specifically, the $(n,k)$-th entry in the resulting sparsified channel matrix $\widehat{\mathbf{H}}$ is given by\\
\begin{equation}
{\widehat  H}_{n,k} =\begin{cases}{H}_{n,k},&d_{n,k}<d_0
\\0, &\text{otherwise,}
\end{cases}
\label{eqn:sparse}
\end{equation}
where $d_0$ is a distance threshold. Given the sparsified channel matrix $\widehat{\mathbf{H}}$, the received signal can be represented as
\begin{equation}
\mathbf{y}=P^{\frac{1}{2}}\widehat{\mathbf{H}}\mathbf{x}+P^{\frac{1}{2}}\widetilde{\mathbf{H}}\mathbf{x}+\mathbf{n},
\end{equation} 
where $\widetilde{\mathbf{H}}=\mathbf{H}-\widehat{\mathbf{H}}.$ The MMSE estimator of $\mathbf{x}$ is approximated by
\begin{equation}
\widehat{\mathbf{x}}\approx P^{\frac{1}{2}}\widehat{\mathbf{H}}^H(P\widehat{\mathbf{H}}\widehat{\mathbf{H}}^H +  \widehat N_0\mathbf I)^{-1}\mathbf{y},\label{eqn:app_x}
\end{equation}
with $\widehat N_0=P\mathrm{E}[\sum_{j\neq k}|\widetilde{H}_{n,j}|^2] +N_0$ for arbitrary RRH $n$.  
\par 
As proven in \cite{fan2014dynamic}, the channel matrix can be sparsified without considerably compromising the SINR. The reason is that as the RRHs and users are uniformly distributed over a large area, an RRH can only receive reasonably strong signals from a small number of nearby users, and vice versa. Therefore, the majority of the elements of $\mathbf{H}$ are relatively small in magnitude, and ignoring them in signal detection leads to marginal loss in the overall system performance. Indeed, according to \cite{fan2014dynamic}, when $N$ scales in the same order as $K$, the distance threshold $d_0$ does not increase with the network size to achieve a certain SINR performance. Thus, in this paper, we assume that $d_0$ is a predetermined constant regardless of the network size. This implies that the average number of users connecting to an RRH does not scale with the network size.
\subsection{Bipartite Random Geometric Graph}
Channel sparsification simplifies the signal detection in a C-RAN to an inference problem over a bipartite random geometric graph (see Fig.~\ref{fig:sfg}). In the bipartite random geometric graph, RRHs and users in a C-RAN are referred to RRH nodes and user nodes respectively, and edge connections exist only between RRH nodes and user nodes. More specifically, an RRH node is connected to a user node only if the distance between them falls below the threshold $d_0$, and the weight over such an edge is the channel coefficient from the corresponding user to the corresponding RRH.
\par 
Suppose that the entries in $\mathbf{x}$ follow an independent complex Gaussian distribution.\footnote{If $\mathbf{x}$ does not follow a Gaussian distribution, the message-passing algorithm presented in this paper gives an approximation of the linear MMSE estimation \cite{weiss2001correctness}.} Then, $\mathbf{y}$ and $\mathbf{x}$ are jointly Gaussian, and therefore the MMSE detector in (\ref{eqn:x}) is also the maximum \textit{a posteriori} probability (MAP) detector that maximizes the \textit{a posteriori} probability $p(\mathbf{x}|\mathbf{y})$ \cite{Vaseghi2008advanced}. That is,
\begin{equation}
\widehat{\mathbf{x}}= \arg \max p(\mathbf{x}|\mathbf{y}).
\label{eqn:xx}
\end{equation}
The probability density function $p(\mathbf{x}|\mathbf{y})$ can be factorized as
\begin{equation}
\begin{aligned}
p(\mathbf{x}|\mathbf{y}) &\propto &&p(\mathbf{y}|\mathbf{x})p(\mathbf{x})
\\
&=&&p(y_1|\mathbf{x})\cdots p(y_n|\mathbf{x})\cdots p(y_N|\mathbf{x})
\\
& 
&&\times p(x_1)\cdots p(x_k)\cdots p(x_K).
\end{aligned}
\label{eqn:pdf_full}
\end{equation}
Recall that we sparsify the channel matrix by using the channel sparsification approach given in \cite{fan2014dynamic}. Based on (\ref{eqn:app_x}), the factorization of $p(\mathbf{x}|\mathbf{y})$ is approximated as 
\begin{equation}
\begin{aligned}
p(\mathbf{x}|\mathbf{y}) &\propto &&p(\mathbf{y}|\mathbf{x})p(\mathbf{x})
\\
&\approx &&p(y_1|\mathbf{x}_{\mathcal{I}_1})\cdots p(y_n|\mathbf{x}_{\mathcal{I}_n})\cdots p(y_N|\mathbf{x}_{\mathcal{I}_N})
\\
& 
&&\times p(x_1)\cdots p(x_k)\cdots p(x_K),
\end{aligned}
\label{eqn:pdf}
\end{equation}
where $\mathbf{x}_{\mathcal{I}_n}$ contains all $x_i$ with $i \in \mathcal{I}_n$ and $\mathcal{I}_n$ is the set of user indices with $d_{n,k}<d_0$.
\begin{figure}[!h]
\centering
{\includegraphics[width=0.4\textwidth]{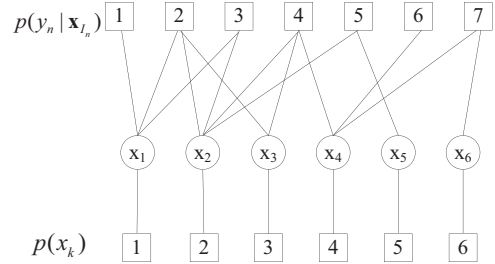}}
\caption{A factor graph corresponding to the C-RAN in Fig.~\ref{fig:sfg}.}\label{fig:sfg_2}
\end{figure}
\par 
We now transfer the bipartite random geometric graph to a factor graph with the factorization in (\ref{eqn:pdf}). As illustrated in Fig.~\ref{fig:sfg_2}, a factor graph is also a bipartite graph comprising two types of nodes, namely, variable nodes (denoted by circles) and check nodes (denoted by squares), together with edges connecting these two types of nodes. The relation between the factorization (\ref{eqn:pdf}) and its associated factor graph is as follows. A check node $p(y_n|\mathbf{x}_{\mathcal{I}_n})$ is connected to a variable node $x_k$ by an edge when there is an edge connecting the $n$-th RRH node and the $k$-th user node in the corresponding random geometric graph (i.e., $d_{n,k}<d_0$), or equivalently, when the function $p(y_n|\mathbf{x}_{\mathcal{I}_n})$ takes $x_k$ as input.
\subsection{Gaussian Message Passing}
\begin{algorithm}[h]
\caption{Gaussian Message-Passing (GMP) Algorithm}
\label{alg::GMP}
{
\begin{algorithmic}[1]
\REQUIRE	
$\widehat{\mathbf{H}}$, $\mathbf{y}$
\ENSURE
$\widehat{x}_k$ for all $k$
\STATE Initial $t = 0, m_{x_k\rightarrow y_n}^{(0)}=0$, $v_{x_k\rightarrow y_n}^{(0)}=1,$ for all $k,n$
\STATE \textbf{Repeat}
\STATE Set $t \Leftarrow t+1$
\STATE For all $n,k$ such that $\widehat{{H}}_{n,k} \neq 0$, compute
\begin{flalign}
&
\small
v_{y_n\rightarrow x_k}^{(t)}=\frac{1}{P{|\widehat{{H}}_{n,k}|^2}}\bigg(\widehat{{N}}_0+P\sum_{j\neq k} |\widehat{H}_{n,j}|^2v_{x_j\rightarrow y_n}^{(t-1)}\bigg)
\label{eqn:v1}
&
\end{flalign}

\begin{flalign}
\small
&m_{y_n\rightarrow x_k}^{(t)}=\frac{1}{P^{\frac{1}{2}}{\widehat{{H}}_{n,k}}}\bigg(y_n-P^{\frac{1}{2}}\sum_{j\neq k} \widehat{{H}}_{n,j}m_{x_j\rightarrow y_n}^{(t-1)}\bigg)
&
\label{eqn:m1}
\end{flalign}
\begin{flalign}
\small
&v_{x_k\rightarrow y_n}^{(t)}=\bigg( \sum_{\widehat{{H}}_{j,k}\neq 0, j\neq n} \frac{1}{v_{y_j\rightarrow x_k}^{(t)}}+1\bigg)^{-1}
\label{eqn:v2}
&
\end{flalign}

\begin{flalign}
\small
&m_{x_k\rightarrow y_n}^{(t)}=v_{x_k\rightarrow y_n}^{(t)} \sum_{\widehat{{H}}_{j,k}\neq 0, j\neq n} \frac{{m_{y_j\rightarrow x_k}^{(t)}}}{{v_{y_j\rightarrow x_k}^{(t)}}}
\label{eqn:m2}
&
\end{flalign}

\STATE \textbf{Until }{\text{the stopping criterion is satisfied}}
\STATE Compute
\begin{flalign}
\small
&v_k=\bigg( \sum_{\widehat{{H}}_{n,k}\neq 0} \frac{1}{v_{y_n\rightarrow x_k}^{(t)}}+1\bigg)^{-1}
&
\end{flalign}

\begin{flalign}
\small
&\widehat{x}_k=v_k\sum_{\widehat{{H}}_{n,k}\neq 0} \frac{{m_{y_n\rightarrow x_k}^{(t)}}}{v_{y_n\rightarrow x_k}^{(t)}}.
&
\end{flalign}
\end{algorithmic}}
\end{algorithm}
We are now ready to introduce the Gaussian message-passing algorithm for signal detection. The algorithm will be implemented in the centralized data center. The messages, namely, the marginals of $\{x_k\}$ and $\{y_n\}$, are exchanged along the edges. In this paper, both $\{x_k\}$ and $\{y_n\}$ are Gaussian distributed, and therefore the messages are Gaussian probability density functions and can be completely characterised by mean and variance. Denote by $m^{(t)}_{y_n\rightarrow x_k}$ and $v^{(t)}_{y_n\rightarrow x_k}$ the mean and variance sent from check node $p(y_n|\mathbf{x}_{\mathcal{I}_n})$ to variable node $x_k$ at iteration $t$, respectively, and denote by $m^{(t)}_{x_k\rightarrow y_n}$ and $v^{(t)}_{x_k\rightarrow y_n}$ the mean and variance sent from variable node $x_k$ to check node $p(y_n|\mathbf{x}_{\mathcal{I}_n})$ at iteration $t$, respectively. The detailed steps of message passing are presented in Algorithm 1. We refer to this algorithm as Gaussian message passing (GMP), as all the messages involved are Gaussian marginals. Note that each RRH only serves users located in a circle with a constant radius $d_0$. Thus, the average number of messages to be exchanged and computed at each node does not scale with the network size. Therefore, the complexity per iteration of the GMP algorithm is linear in the number of RRHs and users.
\par 
In spite of its linear complexity per iteration, the GMP algorithm is not guaranteed to converge on the factor graphs induced by C-RANs. It is known that the GMP algorithm always converges to the optimal solution on a tree-type factor graph\footnote{ A tree-type graph is an undirected graph in which any two nodes are connected by exactly one path, where a path is a sequence of edges which connect a sequence of vertices without repetition.}  \cite{kschischang2001factor}. It is also known that, if a factor graph is random and sparse enough, the corresponding message-passing algorithm converges asymptotically as the network size grows to infinity \cite{richardson2001design}. However, the factor graph for a bipartite random geometric graph induced from a C-RAN is locally dense and far from being a tree. This is due to the fact that every RRH needs to simultaneously serve multiple nearby users. For example, $\{\text{R} 2, \text{U} 1, \text{R} 3, \text{U} 2\}$ in Fig.~\ref{fig:sfg} form a loop \footnote{A loop in a graph is a path that starts and ends at the same node.} of length 4. Indeed, we observe in simulations that the GMP algorithm diverges in C-RAN with a non-trivial probability. Even worse, the probability of divergence grows with the network size, as illustrated later in Fig.~\ref{fig:prob_con}. We focus on improving the convergence performance of GMP in the rest of the paper.

\begin{remark}
The GMP algorithm for a C-RAN with channel sparsification can be simply extended to the case without channel sparsification by setting the distance threshold to infinity. However, this leads to an increase of the computational complexity per iteration. We see that in each iteration of Algorithm 1, messages need to be updated on every edge of the factor graph. From channel randomness, the entries of $\mathbf{H}$ are non-zero with probability one. Thus, in the factor graph without channel sparsification, every RRH check node $p(y_n|\mathbf{x})$ is connected to all variable nodes $\{x_k\}_{k=1}^{K}$. This implies that the total number of edges in the factor graph is $NK$, implying that the complexity of the GMP algorithm is $O(NK)$ per iteration, which is unaffordable for a large-scale C-RAN. 
\end{remark}
\section{Randomized Gaussian Message Passing with Channel Sparsification}

\begin{algorithm}[h]
\caption{Randomized Gaussian Message-Passing (RGMP) Algorithm}
\label{alg::conjugateGradient}
{
\begin{algorithmic}[1]
\REQUIRE	
$\widehat{\mathbf{H}}$, $\mathbf{y}$
\ENSURE
$\widehat{x}_k$ for all $k$
\STATE initialize $t = 0, m_{x_k\rightarrow y_n}^{(0)}=0$, $v_{x_k\rightarrow y_n}^{(0)}=1,$ for all $k,n$.
\STATE \textbf{Repeat}
\STATE Set $t \Leftarrow t+1$.
\STATE Generate $K$ random variables $\sigma_t(1),\cdots, \sigma_t(K)$ from a continuous uniform distribution on interval $(0, B)$.
\STATE For $i=1,\cdots, K$, at time $\sigma_t(i)$, compute
\begin{flalign}
\small
\small
&v_{y_n\rightarrow x_{i}}^{(t)}
=&&\frac{1}{P{|\widehat{H}_{n,i}|^2}}\bigg(\widehat{N}_0+P\sum_{j:\sigma_t(j)< \sigma_t(i)} |\widehat{H}_{n,j}|^2 v_{x_{j)}\rightarrow y_n}^{(t)}\nonumber
\\
& &&+P\sum_{j\neq i:\sigma_t(j)\geq\sigma_t(i)} |\widehat{H}_{n,j}|^2v_{x_{j}\rightarrow y_n}^{(t-1)}\bigg)
\label{eqn:rv1}
&
\end{flalign}
\begin{flalign}
\small
&m_{y_n\rightarrow x_{i}}^{(t)}
=&&\frac{1}{P^{\frac{1}{2}}{\widehat{H}_{n,{i}}}}\bigg(y_n-P^{\frac{1}{2}}\sum_{j:\sigma_t(j)< \sigma_t(i)} \widehat{H}_{n,j}
m_{x_{j}\rightarrow y_n}^{(t)}\nonumber
\\
& &&-P^{\frac{1}{2}}\sum_{j\neq i:\sigma_t(j)\geq \sigma_t(i)} \widehat{H}_{n,j}m_{x_{j}\rightarrow y_n}^{(t-1)}\bigg)
&\label{eqn:rm1}
\end{flalign}
\begin{flalign}
\small
&v_{x_{i}\rightarrow y_n}^{(t)}=\bigg( \sum_{\widehat{H}_{j,{i}}\neq 0, j\neq n} \frac{1}{v_{y_j\rightarrow x_{i}}^{(t)}}+1\bigg)^{-1}
&\label{eqn:rv2}
\end{flalign}
\begin{flalign}
\small
&m_{x_{i}\rightarrow y_n}^{(t)}
=v_{x_{i}\rightarrow y_n}^{(t)} \sum_{\widehat{H}_{j,{i}}\neq 0, j\neq n} \frac{{m_{y_j\rightarrow x_{i}}^{(t)}}}{{v_{y_j\rightarrow x_{i}}^{(t)}}}
&\label{eqn:rm2}
\end{flalign}
\STATE \textbf{Until }{\text{stopping criteria is satisfied}}
\STATE Compute 
\begin{flalign}
\small
&v_k=\bigg( \sum_{\widehat{H}_{n,k}\neq 0} \frac{1}{v_{y_n\rightarrow x_k}^{(t)}}+1\bigg)^{-1}&
\end{flalign}
\begin{flalign}
\small
&\widehat{x}_k=v_k\sum_{\widehat{H}_{n,k}\neq 0} \frac{{m_{y_n\rightarrow x_k}^{(t)}}}{{v_{y_n\rightarrow x_k}^{(t)}}}.&
\end{flalign}
\end{algorithmic}}
\end{algorithm}
\begin{figure*}[b]
\begin{equation}
\small
\mathbf{H}
=10^{-5}
\left[ \begin{array}{cccc}  
-0.1458 + 0.2401i& -2.0998 - 0.7353i&  -2.1459 - 2.0284i&   0.6130 + 2.0420i\\
  17.7199 +18.8315i&   1.8431 - 2.4183i&   5.7441 + 2.0536i&   0.4837 - 3.0383i\\
     5.1714 -14.5292i&   0.1184 - 1.5314i& -10.3012 + 0.1049i&   2.4388 - 0.8546i\\
 -25.2041 -16.2758i&   1.1697 - 0.3792i&   2.2858 - 0.2858i&   6.0425 - 2.6317i
 \end{array} \right]\label{eqn:H}
\end{equation}
\end{figure*}
\subsection{Randomized Gaussian Message Passing}
In this section, we propose the RGMP algorithm to address the convergence issue of GMP. The main novelty of the RGMP algorithm is on the scheduling strategy for message updating. The conventional GMP algorithm employs synchronous message passing, i.e., messages are updated in parallel. As aforementioned, synchronous message passing does not work well in C-RANs due to local loops in the factor graph. It is well-known that serial message passing improves the convergence performance of GMP \cite{Goldberger2008serial}. As shown in our later simulations, the convergence of GMP heavily depends on the update order. Nonetheless, there is no systematic way to derive a fixed update order that guarantees convergence. To address this issue, we propose the RGMP algorith, which updates messages in a random order.

\par 
The RGMP algorithm is described as follows. At the $t$-th iteration, $K$ random variables $\sigma_t(k)$ are generated from a continuous uniform distribution $(0,B)$. Then, the messages at the variable node $x_k$ are updated at the time $\sigma_t(k)$. For example, when $K=3$, $B=1$, at the $t$-th iteration with $\sigma_t(1)=0.623$, $\sigma_t(2)=0.307$, $\sigma_t(3)=0.890$, we first update all the messages on the edges connecting the variable node $x_2$ at time $0.307$. Then, the messages on the edges connecting the variable node $x_1$ are updated at time $0.623$. Finally, messages related to variable node $x_3$ are updated at time $0.890$. The detailed RGMP algorithm is given in Algorithm 2. 

\begin{remark}
In an ideal case when the updating and exchanging of messages does not incur any delay, each message can be updated with the up-to-date information. In this way, the RGMP algorithm updates the messages sequentially in a randomly permuted order. This is because the update time of messages is randomly generated continuous variables. With probability one, messages related to different variable nodes will not be updated at the same time. In reality, however, message updating and exchanging may cause a non-negligible delay, which means messages may not be updated with the latest information. If the update time interval $B$ is much smaller than the delay, the RGMP algorithm is equivalent to the synchronous GMP. To make the RGMP algorithm different from synchronous GMP, the time interval $B$ should be comparable to the overall delay in each iteration. Consequently, the computation time of RGMP will be significantly increased. Moreover, generating a large number of continuous random variables for creating the random schedule also introduces a non-negligible computational complexity. In Section VI, we will introduce a blockwise RGMP (B-RGMP) algorithm for parallel implementation. The total computation time of B-RGMP will not increase with the network size.
\end{remark}

\subsection{Numerical Examples}
\begin{figure*}[!h]
\centering
\subfigure[RGMP with $B=1$]
{\includegraphics[width=0.38\textwidth]{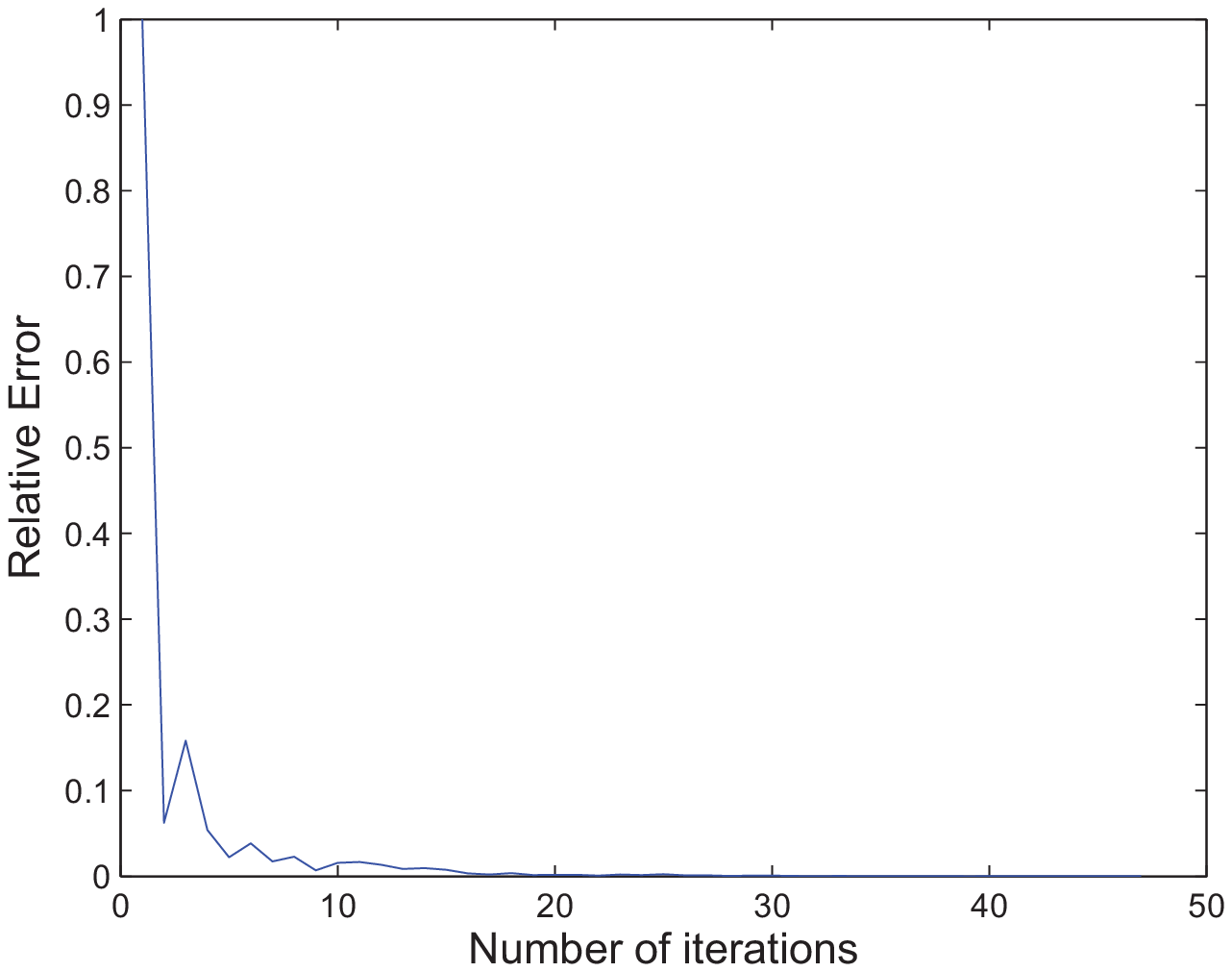}
}
\hspace{0.4in}
\subfigure[Synchronous Gaussian message passing]
{\includegraphics[width=0.38\textwidth]{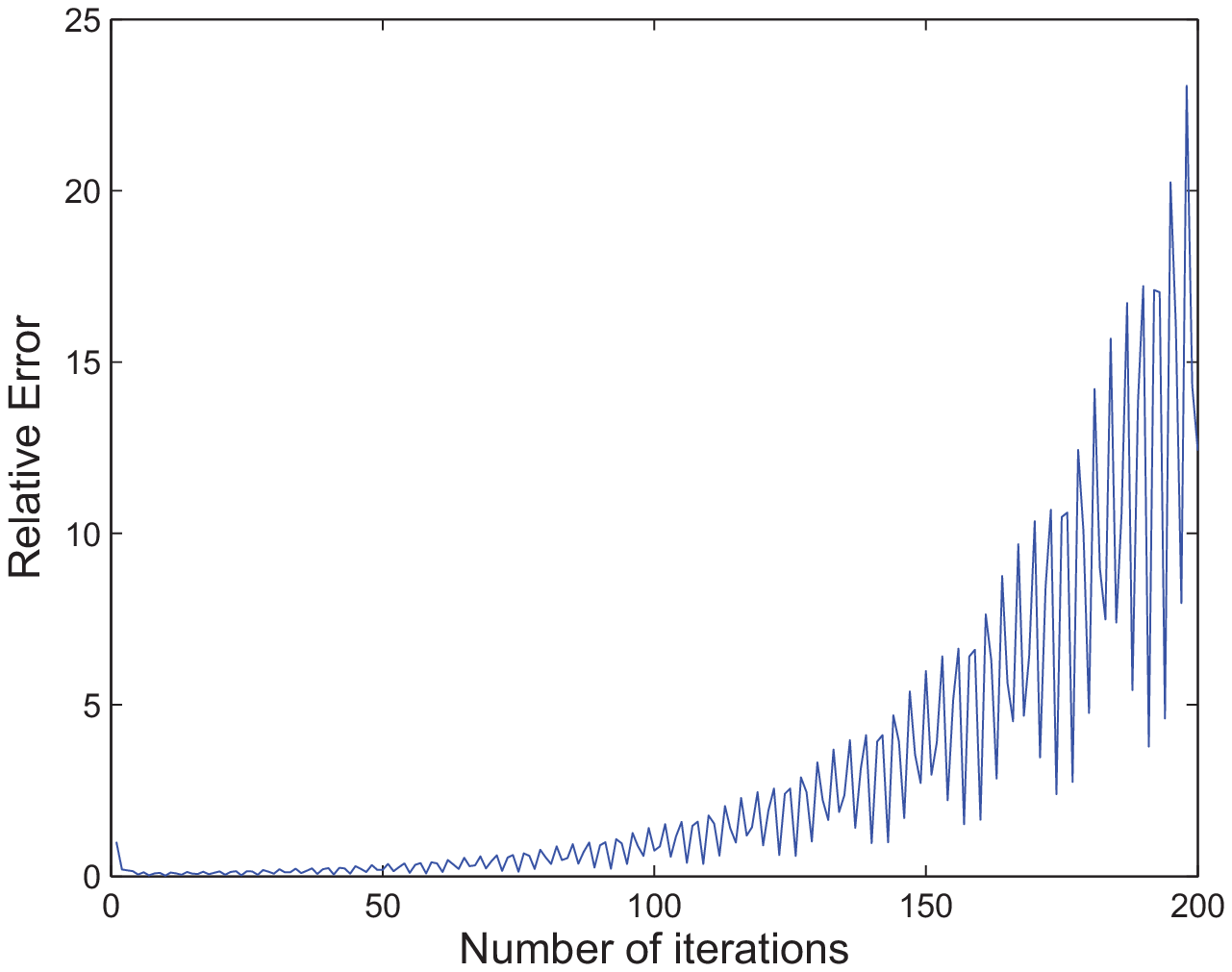}}
\hspace{0.4in}
\subfigure[Asynchronous Gaussian message passing with schedule $\sigma=(0.198,0.432,0.909,0.859)$]
{\includegraphics[width=0.38\textwidth]{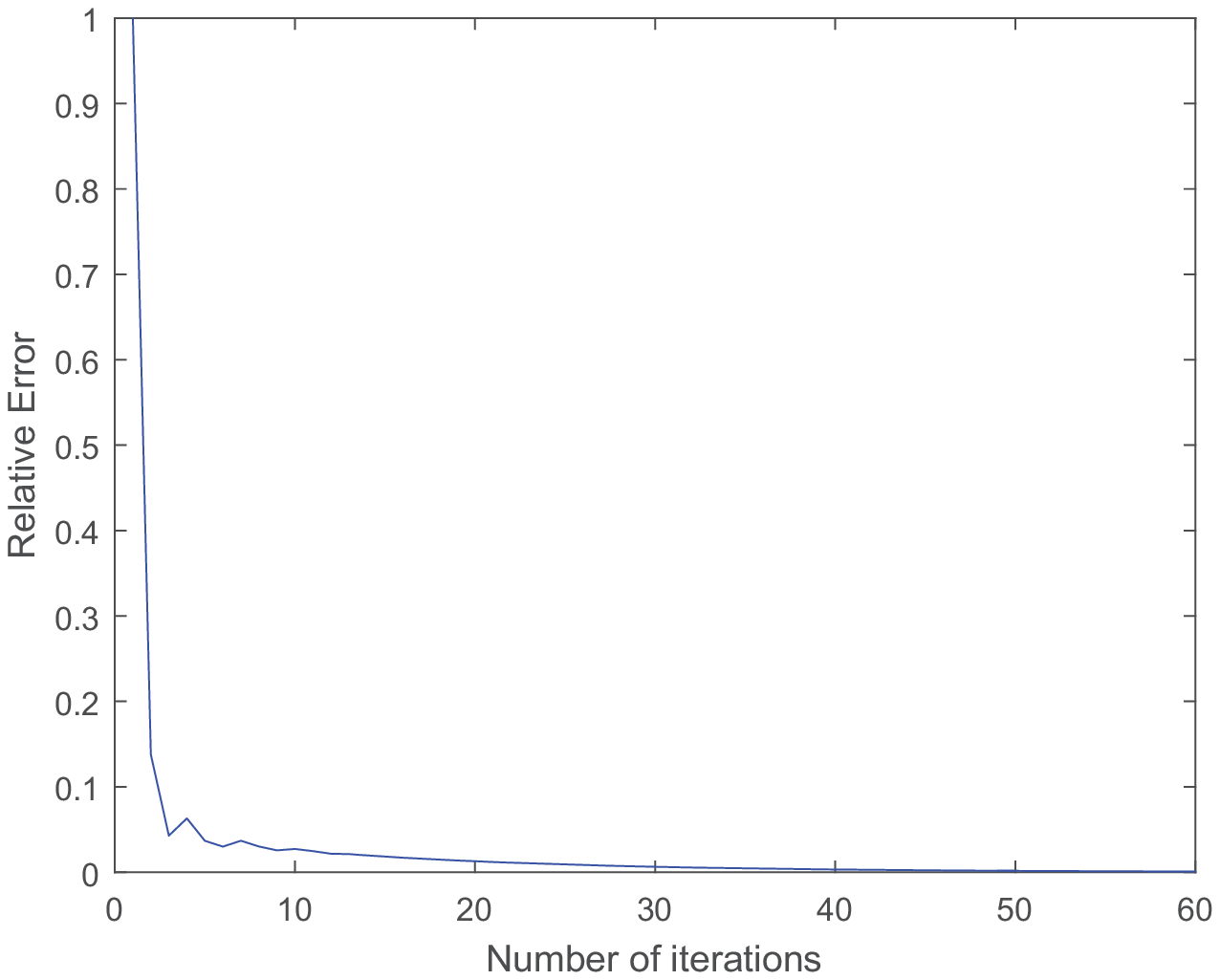}\label{fig:err_mpc}}
\hspace{0.4in}
\subfigure[Asynchronous Gaussian message passing with schedule $\sigma=(0.198,0.432,0.859,0.909)$]
{\includegraphics[width=0.38\textwidth]{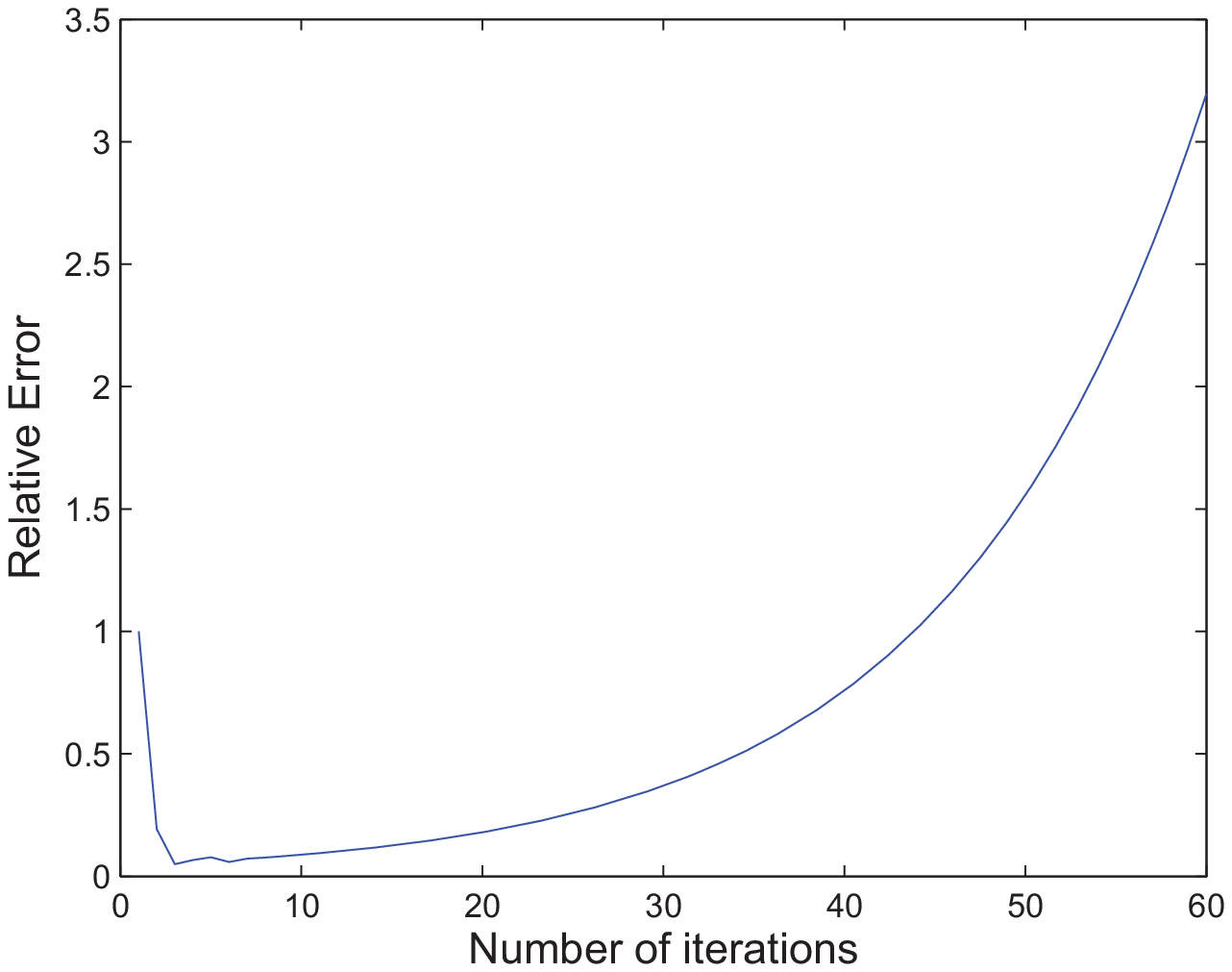}\label{fig:err_mpd}}
\caption{Relative error vs number of iterations.}\label{fig:err_mp}
\end{figure*}
In this subsection, we use a simple example to illustrate the difference between our proposed RGMP algorithm and synchronous/asynchronous GMP. Consider the randomly generated channel matrix in (\ref{eqn:H}) and let the transmit SNR (i.e., $\frac{P}{N_0}$) be $100$dB. The corresponding received signal $\mathbf{y}$ is 
\begin{equation}
\small
\mathbf{y} = \begin{bmatrix}
     1.6847 - 7.1280i\\
 -20.9794 + 3.6052i\\
  -3.0214 + 3.8041i\\
  21.5306 + 6.5308i
\end{bmatrix}
\label{eqn:y}
\end{equation}
For fairness of comparison, we do not conduct channel sparsification in this example. That is, the distance threshold is set to infinity. Fig.~\ref{fig:err_mp} plots the relative error versus the number of iterations for the RGMP algorithm with $B=1$ and the GMP algorithm with different message updating strategies, i.e., synchronous update and asynchronous update with different fixed schedules, $\sigma=(0.198,0.432,0.909,0.859)$ and $\sigma=(0.198,0.432,0.859,0.909)$. The relative error is defined as $\frac{\| P\widehat{\mathbf{H}}^H\widehat{\mathbf{H}}\mathbf{x}^{(t)}-{P}^{\frac{1}{2}}\widehat{\mathbf{H}}^H\mathbf{y}\|}{\|{P}^{\frac{1}{2}}\widehat{\mathbf{H}}^H\mathbf{y}\|}$, where $\mathbf{x}^{(t)}$ is the estimation of the transmitted signal after the $t$-th iteration. We see that the synchronous GMP algorithm and the asynchronous one with schedule $(0.198,0.432,0.859,0.909)$ diverge, but the asynchronous GMP with schedule $(0.198,0.432,0.909,0.859)$ and the proposed RGMP algorithm converge.
\begin{remark}
The examples in Fig.~\ref{fig:err_mpc} and \ref{fig:err_mpd} show that convergence of asynchronous GMP heavily depends on the update schedule/order. Unfortunately, there is no systematic way to derive a fixed update order that guarantees convergence. In general, finding such an update order is difficult, especially in large networks. This issue is avoided in the proposed RGMP algorithm by randomizing the update schedule instead of fixing one. Indeed, the randomization significantly weakens the loopy effect of the graph, and thus convergence is almost ensured in RGMP. 
\end{remark}
   \begin{figure}[!htb]
	\centering
	{\includegraphics[width=0.48\textwidth]{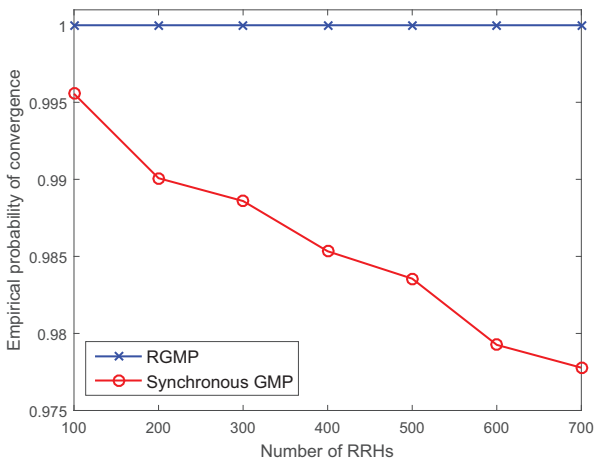}}
	\caption{Probability of convergence with $\beta_N=10/\text{km}^2$, $\beta_K=8/\text{km}^2$, and $P=95$dB.}\label{fig:prob_con}
\end{figure}
\par 
In Fig.~\ref{fig:prob_con}, we plot the empirical probability of convergence against the network size, where users and RRHs are uniformly located in a circular network area with user density $8/\text{km}^2$ and RRH density $10/\text{km}^2$. The distance threshold $d_0$ is $1000$m. For each simulated point in Fig.~\ref{fig:prob_con}, both GMP and RGMP are run for over $6000$ times that are randomized over both RRH/user location and channel fading. For GMP, the convergence probability decreases when the network size becomes large. In contrast, no divergence has been observed for the RGMP algorithm throughout our simulations.
\section{Convergence Analysis}
It is proven that the fixed point of GMP always provides the exact marginals (i.e., the solution of MMSE detection in this paper), provided that the algorithm converges \cite{weiss2001correctness}. Thus, we only need to consider the convergence of the proposed algorithm, since the algorithm always gives the true solution of MMSE detection as long as it converges. In this section, we establish a necessary and sufficient condition for the expected convergence of the proposed RGMP algorithm. For self-containedness, we start with existing results on the analysis of the convergence condition for conventional GMP. 
\subsection{Convergence of GMP}
The factor graph of a C-RAN contains loops with high probability. The convergence of GMP on a loopy factor graph has been previously studied in \cite{ng2007distributed}, with the main result summarized below. 
\par 
From Algorithm 1, we see that the evolution of the variances $v_{y_n\rightarrow x_k}$ is independent of the means $m_{y_n\rightarrow x_k}$, $m_{x_k\rightarrow y_n}$ and the received signal $\mathbf{y}$. Substituting (\ref{eqn:v2}) into (\ref{eqn:v1}), we obtain
\begin{equation}
\begin{aligned}
&v_{y_n\rightarrow x_k}^{(t)}
\\
=&\frac{\widehat{N}_0+P\sum_{j\neq k} |\widehat {H}_{n,j}|^2\bigg( \sum_{\widehat {H}_{i,j}\neq 0, i\neq n} \frac{1}{v_{y_i\rightarrow x_j}^{(t-1)}}+1\bigg)^{-1}}{P{|\widehat {H}_{n,k}|^2}}
\end{aligned}.
\label{eqn:evo_v}
\end{equation}
Denote (\ref{eqn:evo_v}) in a vector form as 
\begin{equation}
\mathbf v^{(t)}=f(\mathbf v^{(t-1)}),
\end{equation}
where $f(\cdot)$ is the evolution function determined by (\ref{eqn:evo_v}), and $\mathbf{v}^{(t)}$ is a vector consisting of $v_{y_n\rightarrow x_k}^{(t)}$ for all $n$ and $k$ with $\widehat H_{n,k}\neq 0$. Note that $f(\cdot)$ is a standard function, the definition of which is given below.
\begin{definition}
A function $f(\mathbf{v})$ is standard if for all $\mathbf{v}\geq \mathbf{0}$ the following properties are satisfied.
\begin{itemize}
\item \textit{Positivity:} $f(\mathbf{v})>\mathbf{0}$.
\item \textit{Monotonicity:} If $\mathbf{v}\geq \mathbf{v}'$, then $f(\mathbf{v})\geq f(\mathbf{v}')$.
\item \textit{Scalability:} For all $\alpha>1$, $\alpha f(\mathbf{v})>f(\alpha \mathbf{v})$.
\end{itemize}
\end{definition}
Furthermore, we prove that the variances of GMP always converge to a unique fixed point in Lemma 1 \footnote{Lemma 1 was previously shown in Theorem 5.1 of \cite{ng2007distributed}, but the proof has been omitted in \cite{ng2007distributed}. Here, we include the detailed proof of Lemma 1 for self-containedness.}.
\begin{lemma}
In the GMP algorithm, if the initial point $\mathbf{v}^{(0)} > \mathbf{0}$, the sequence of $\mathbf{v}^{(t)}$ always converges to a fixed point of $f(\cdot)$ and the fixed point is unique.
\end{lemma}
\begin{proof}
As proven in Theorem 2 of \cite{yates1995framework}, if a standard function has a fixed point and the initial point is positive, the algorithm always converges to a unique fixed point of the standard function. Thus, it suffices to show that $f(\cdot)$ has a fixed point. Since $\widehat{N}_0>0$, we suppose that $\mathbf{0}<\mathbf{v}^{(0)}< \frac{\widehat{N}_0}{P{|\widehat {H}_{n,k}|^2}} \cdot \mathbf{1}$ for all $\widehat {H}_{n,k} \neq 0$. Then, we can see that $\mathbf{v}^{(0)}< \mathbf{v}^{(1)}$. Consequently, we obtain $\mathbf{v}^{(0)}< \mathbf{v}^{(1)}<\cdots < \mathbf{v}^{(t-1)}< \mathbf{v}^{(t)}$. The sequence of variances is an increasing sequence. Moreover, the sequence is upper bounded by $c \cdot \mathbf{1}$, where $c$ satisfies the following conditions
\begin{equation}
c\geq \frac{1}{P{|\widehat {H}_{n,k}|^2}}\bigg(\widehat{N}_0+P\sum_{j\neq k} |\widehat {H}_{n,j}|^2\bigg), \forall \widehat {H}_{n,k} \neq 0.
\end{equation}
Thus, the sequence of variances always converges to a limit point, and the limit point is a fixed point of $f(\mathbf{v})$. This concludes the proof.
\end{proof}
\par 
We now consider the convergence of means. A vector of means, $\mathbf{m}^{(t)}$, is constructed with its $((k-1)N+n)$-th entry being 
\begin{equation}
\begin{aligned}
&m_{(k-1)N+n}^{(t)}
=&&
\begin{cases}
m_{y_n\rightarrow x_k}^{(t)},&\widehat {H}_{n,k}\neq 0,
\\0, &\text{otherwise.}
\end{cases}
\end{aligned}
\end{equation}
The recursion of the means is given by (\ref{eqn:m1}) and (\ref{eqn:m2}). As the variances always converge, the evolution of the means can be written as follows:
\begin{equation}
\mathbf{m}^{(t)}=\mathbf \Omega \mathbf{m}^{(t-1)}+\mathbf{z},
\label{eqn:evo_m}
\end{equation}
where $\mathbf{z}$ is an $NK\times 1$ vector with its $((k-1)N+n)$-th entry being 
\begin{equation}
\begin{aligned}
&z_{(k-1)N+n}
=&&
\begin{cases}
\frac{y_n}{P^{\frac{1}{2}}\widehat{H}_{n,k}},&\widehat {H}_{n,k}\neq 0,
\\0, &\text{otherwise,}
\end{cases}
\end{aligned}
\end{equation}
and $\mathbf \Omega$ is an $NK\times NK$ matrix with the $((k-1)N+n,(j-1)N+i)$-th entry being 
\begin{equation}
\begin{aligned}
&\Omega_{(k-1)N+n,(j-1)N+i}\\
=&
\begin{cases}
-\frac{\widehat {H}_{n,j}v_{x_j\rightarrow y_n}^*}{\widehat {H}_{n,k}v_{y_i\rightarrow x_j}^*},&\widehat {H}_{n,k}\neq 0, \widehat {H}_{i,j}\neq 0, n\neq i,\text{and }j\neq k,
\\0, &\text{otherwise,}
\end{cases}
\end{aligned}
\end{equation}
with $v_{x_j\rightarrow y_n}^*=\left(\sum_{\widehat{H}_{i,j}\neq 0,i\neq n}\frac{1}{v_{y_i\rightarrow x_j}^*}+1\right)^{-1}$ and $v_{y_i\rightarrow x_j}^*=\lim_{t\rightarrow \infty}v_{y_i\rightarrow x_j}^{(t)}$. Then, a necessary and sufficient condition for the convergence of (\ref{eqn:evo_m}) is given in Theorem 5.2, \cite{ng2007distributed}. That is, in Algorithm 1, the sequence of $\mathbf{m}^{(t)}$ converges to a unique fixed point if and only if the spectral radius $\rho(\mathbf \Omega)<1$.

\subsection{Convergence of RGMP}
In this subsection, we first show that the message variances always converge in the RGMP algorithm. Then, we focus on the convergence condition of the means in RGMP.
\par 
Recall that the evolution function of the variances in (\ref{eqn:evo_v}) is a standard function with a unique fixed point. As proven in \cite{yates1995framework}, if the evolution function of a synchronous algorithm is standard and feasible, then the corresponding asynchronous algorithm converges. Based on that, we obtain the following theorem.
\begin{theorem}
In the RGMP algorithm, the sequence of $v_{y_n\rightarrow x_k}^{(t)}$ always converges to the same unique fixed point as in Algorithm 1 if the initial point $\mathbf{v}>\mathbf{0}$.
\end{theorem}
\par 
With Theorem 1, it suffices to focus on the convergence condition of the means in the RGMP algorithm. Denote the update schedule at the $t$-th iteration as $\sigma_t$. Combining (\ref{eqn:rm1}) and (\ref{eqn:rm2}), we obtain the evolution of means $\mathbf{m}^{(t)}$ as\\
\begin{equation}
\mathbf{m}_{k}^{(t+1)}=\sum_{j:\sigma_t(j)<\sigma_t(k)} \mathbf \Omega_{k,j}\mathbf{m}_{j}^{(t+1)}+\sum_{j:\sigma_t(j)\geq \sigma_t(k)} \mathbf \Omega_{k,j}\mathbf{m}_{j}^{(t)}+\mathbf{z}_{k},
\label{eqn:random_mp}
\end{equation} 
where $\mathbf{m}^{(t)}_j$ is an $N\times 1$ subvector of $\mathbf{m}^{(t)}$ with the $n$-th entry being
\begin{equation}
m_j^{(t)}(n)=
\begin{cases}
m^{(t)}_{y_n\rightarrow x_j},&\widehat{H}_{n,j}\neq 0,
\\0, &\text{otherwise,}
\end{cases}
\end{equation}
and $\mathbf{z}_{k}$ is an $N\times 1$ subvector of $\mathbf{z}$ with the $n$-th entry being
\begin{equation}
z_k(n)=
\begin{cases}
\frac{y_n}{P^{\frac{1}{2}}\widehat{H}_{n,k}},&\widehat{H}_{n,k}\neq 0,
\\0, &\text{otherwise.}
\end{cases}
\end{equation}
$\mathbf \Omega_{k,j}$ is the $N\times N$ evolution matrix from user $j$ to user $k$ with the $(n,m)$-th entry being
\begin{equation}
\begin{aligned}
&\mathbf{\Omega}_{k,j}(n,m)\\
=&
\begin{cases}
-\frac{\widehat {H}_{n,j}v_{x_{j}\rightarrow y_n}^*}{\widehat {H}_{n,k}v^*_{y_m\rightarrow x_{j}}},&\widehat {H}_{n,k}\neq 0, \text{and } \widehat {H}_{n,j}\neq 0, \text{and } n\neq m,
\\0, &\text{otherwise.}
\end{cases}
\end{aligned}
\end{equation}
$\mathbf \Omega_{k,j}$ is the $(k,j)$-th submatrix of $\mathbf{\Omega}$. More specifically, 
\begin{equation}
\mathbf{\Omega}=
\begin{bmatrix}
     \mathbf{0} &\mathbf{\Omega}_{1,2} &\cdots&\mathbf{\Omega}_{1,K-1}&\mathbf{\Omega}_{1,K} \\[0.6em]
      \mathbf{\Omega}_{2,1} &\mathbf{0} &\mathbf{\Omega}_{2,3} &\cdots &\mathbf{\Omega}_{2,K}\\[0.6em]
    \vdots    &\ddots &\ddots & \ddots&\vdots\\[0.6em]
       \mathbf{\Omega}_{K-1,1} &\cdots &\mathbf{\Omega}_{K-1,K-2} &\mathbf{0}&\mathbf{\Omega}_{K-1,K}\\[0.6em]
       \mathbf{\Omega}_{K,1}& \mathbf{\Omega}_{K,2}&\cdots&\mathbf{\Omega}_{K-1,K}&\mathbf{0}
     \end{bmatrix}.
\end{equation}
\par 
We can further rewrite the equation (\ref{eqn:random_mp}) as
\begin{equation}
\mathbf{L}_{\sigma_t}\mathbf{m}^{(t+1)}=\mathbf{R}_{\sigma_t}\mathbf{m}^{(t)}+\mathbf{z},
\label{eqn:evo_m_r}
\end{equation}
where $\mathbf{L}_{\sigma_t}=[\mathbf{L}_{\sigma_t}(k,j)]_{k,j}\in \mathbb{C}^{NK\times NK}$ with its $(k,j)$-th submatrix being
\begin{equation}
\mathbf{L}_{\sigma_t}(k,j)=\begin{cases}
-\mathbf \Omega_{k,j},&\sigma_t(k)>\sigma_t(j)\\
\mathbf{I},&k=j\\
\mathbf{0}, &\text{otherwise,}
	\end{cases}
	\label{eqn:L1}
\end{equation}
and $\mathbf{R}_{\sigma_t}=[\mathbf{R}_{\sigma_t}(k,j)]_{k,j}\in \mathbb{C}^{NK\times NK}$ with its $(k,j)$-th submatrix being
\begin{equation}
\mathbf{R}_{\sigma_t}(k,j)=\begin{cases}
\mathbf \Omega_{k,j},&\sigma_t(k)\leq \sigma_t(j)\\
\mathbf{0}, &\text{otherwise.}
	\end{cases}
	\label{eqn:R1}
\end{equation}
Based on the definition of $\mathbf{L}_{\sigma_t}$, the determinant of $\mathbf{L}_{\sigma_t}$ is always $1$ or $-1$. It implies that $\mathbf{L}_{\sigma_t}$ is nonsingular. Then, multiplying both sides of (\ref{eqn:evo_m_r}) by $\mathbf{L}_{\sigma_t}^{-1}$, we obtain
\begin{equation}
\mathbf{m}^{(t+1)}=\mathbf{L}_{\sigma_t}^{-1}\mathbf{R}_{\sigma_t}\mathbf{m}^{(t)}+\mathbf{L}_{\sigma_t}^{-1}\mathbf{z}.
\label{eqn:update_1}
\end{equation}
Consequently, we obtain the following condition for the convergence of the RGMP algorithm.
\begin{proposition}
For a given sequence of update schedules $(\sigma_1,\sigma_2,\cdots,\sigma_t,\cdots)$, the RGMP algorithm converges to the fixed point $(\mathbf{I}-\mathbf{\Omega})^{-1}\mathbf{z}$ if and only if $\lim_{t\rightarrow \infty}\mathbf{L}_{\sigma_t}^{-1}\mathbf{R}_{\sigma_t}\cdots\mathbf{L}_{\sigma_1}^{-1}\mathbf{R}_{\sigma_1}=\mathbf{0}$, where $\sigma_t$ is the update schedule at the $t$-th iteration.
\end{proposition}
\begin{proof}
For an arbitrary update schedule $\sigma$, the fixed point of
\begin{equation}
\mathbf{m}^{(t+1)}=\mathbf{L}_{\sigma}^{-1}\mathbf{R}_{\sigma}\mathbf{m}^{(t)}+\mathbf{L}_{\sigma}^{-1}\mathbf{z}
\label{eqn:update_2}
\end{equation}
is given by
\begin{equation}
\mathbf{m}^*=(\mathbf{I}-\mathbf{L}_{\sigma}^{-1}\mathbf{R}_{\sigma})^{-1}\mathbf{L}_{\sigma}^{-1}\mathbf{z}.
\label{eqn:update_3}
\end{equation}
Substituting $\mathbf{R}_{\sigma}=\mathbf{\Omega}-\mathbf{I}+\mathbf{L}_{\sigma}$ into (\ref{eqn:update_3}), we obtain $\mathbf{m}^*=(\mathbf{I}-\mathbf{\Omega})^{-1}\mathbf{z}$. Clearly, $\mathbf{m}^*$ is independent of the choice of schedule $\sigma$. Define $\mathbf{e}^{(t)}=\mathbf{m}^{(t)}-(\mathbf{I}-\mathbf{\Omega})^{-1}\mathbf{z}$. Then, $\mathbf{e}^{(t)}=\mathbf{L}_{\sigma_t}^{-1}\mathbf{R}_{\sigma_t}\mathbf{e}^{(t-1)}$, for any iteration number $t$. By recursion, we obtain $\mathbf{e}^{(t+1)}=\mathbf{L}_{\sigma_t}^{-1}\mathbf{R}_{\sigma_t}\cdots\mathbf{L}_{\sigma_1}^{-1}\mathbf{R}_{\sigma_1}\mathbf{e}^{(0)}$. Therefore, $\mathbf{m}^{(t)} \rightarrow \mathbf{m}^*$ provided $\mathbf{L}_{\sigma_t}^{-1}\mathbf{R}_{\sigma_t}\cdots\mathbf{L}_{\sigma_1}^{-1}\mathbf{R}_{\sigma_1} \rightarrow \mathbf{0}$ as $t\rightarrow \infty$. This concludes the proof.
\end{proof}
Proposition 2 discusses the convergence condition for a given sequence of update schedules $(\sigma_1,\sigma_2,\cdots,\sigma_t,\cdots)$. To quantify the average performance of RGMP over random update schedules, we consider expected convergence in the following, where the expectation is taken over all possible schedules. Let the expected output be 
\begin{equation}
\bm \phi^{(t)}=\mathrm{E}_{\xi_{t-1}}[\mathbf{m}^{(t)}],
\end{equation}
where $\xi_t=(\sigma_1,\cdots,\sigma_t)$ is the set of the update schedules after iteration $t$. We are now ready to present a necessary and sufficient condition for the convergence of $\bm{\phi}^{(t)}$.
\begin{theorem}
The expected output $\bm \phi^{(t)}=\mathrm{E}_{\xi_{t-1}}[\mathbf m^{(t)}]$ converges to the unique point $(\mathbf{I}-\mathbf{\Omega})^{-1}\mathbf{z}$ if and only if the spectral radius $\rho(\mathbf \Lambda)<1$, where
\begin{equation}
\mathbf \Lambda\triangleq\mathrm{E}_{\sigma}[\mathbf L_{\sigma}^{-1}\mathbf R_{\sigma}]
.\label{eqn:condition}
\end{equation}
\end{theorem}
\begin{proof}
Denote $\mathbf{A}\triangleq\mathrm{E}_{\sigma}[\mathbf L_{\sigma}^{-1}]$. Based on the definition of $\phi^t$, we obtain
\begin{equation}
\begin{aligned}
\bm \phi^{(t+1)}&=&&\mathrm{E}_{\xi_{t}}[\mathbf L_{\sigma_t}^{-1}\mathbf R_{\sigma_t}\mathbf{m}^{(t)}+\mathbf{L}_{\sigma_t}^{-1}\mathbf{z}]\\
&=&&\mathrm{E}_{\sigma_t}\left[\mathrm{E}_{\xi_{t-1}}\left[\mathbf L_{\sigma_t}^{-1}\mathbf R_{\sigma_t}\mathbf{m}^{(t)}+\mathbf{L}_{\sigma_t}^{-1}\mathbf{z}\right]\right]\\
&=&&\mathbf{\Lambda}\bm \phi^{(t)}+\mathbf{A}\mathbf{z}.
\end{aligned}
\end{equation}
\par 
From Theorem 5.3 in \cite{axelsson1994iterative}, the sequence of $\bm{\phi}^{(t)}$ converges to the fixed point $(\mathbf{I}-\mathbf{\Lambda})^{-1}\mathbf{A}\mathbf{z}$ if and only if the spectral radius $\rho(\mathbf \Lambda)<1$. Then, it suffices to show that $(\mathbf{I}-\mathbf{\Lambda})^{-1}\mathbf{A}\mathbf{z}=(\mathbf{I}-\mathbf{\Omega})^{-1}\mathbf{z}$.
\par 
Note that $\mathbf{\Omega}=\mathbf{R}_{\sigma}-\mathbf{L}_{\sigma}+\mathbf{I}$. Substituting $\mathbf{R}_{\sigma}=\mathbf{\Omega}+\mathbf{L}_{\sigma}-\mathbf{I}$ into (\ref{eqn:condition}), we obtain
\begin{equation}
\begin{aligned}
\mathbf{\Lambda}&=&&\mathrm{E}_{\sigma}[\mathbf L_{\sigma}^{-1}(\mathbf{\Omega}+\mathbf{L}_{\sigma}-\mathbf{I})]\\
&=&&\mathrm{E}_{\sigma}[\mathbf L_{\sigma}^{-1}(\mathbf{\Omega}-\mathbf{I})+\mathbf{I}]\\
&=&&\mathbf{A}(\mathbf{\Omega}-\mathbf{I})+\mathbf{I}.
\end{aligned}\label{eqn:proofa}
\end{equation}
Recall that $\rho(\mathbf \Lambda)<1$, and thus $\mathbf{I}-\mathbf{\Lambda}$ is nonsingular. Together with $\mathbf{I}-\mathbf{\Lambda}=\mathbf{A}(\mathbf{I}-\mathbf{\Omega})$ from (\ref{eqn:proofa}), we see that both $\mathbf{A}$ and $\mathbf{I}-\mathbf{\Omega}$ are nonsingular. Hence, 
\begin{equation}
(\mathbf{I}-\mathbf{\Lambda})^{-1}\mathbf{A}\mathbf{z}=\left(\mathbf{A}(\mathbf{I}-\mathbf{\Omega})\right)^{-1}\mathbf{A}\mathbf{z}=(\mathbf{I}-\mathbf{\Omega})^{-1}\mathbf{z},
\end{equation} 
which concludes the proof.
\end{proof}
As illustrated later in Fig.~\ref{fig:cdf_rho}, $\rho(\mathbf \Lambda)$ is more likely to take small values than $\rho(\mathbf \Omega)$, which implies that RGMP converges with a higher probability than GMP. Indeed, we have run over $10000$ times, and $\rho(\mathbf \Lambda)<1$ for all cases.
\section{Blockwise RGMP and its Convergence Analysis}
The proposed RGMP algorithm is conceptually simple, but may be cumbersome in implementation, for that the serial message updating schedule prohibits parallel computation. In this section, we generalize RGMP to the blockwise RGMP (B-RGMP) algorithm, which is suitable for parallel message updating. We show that the B-RGMP algorithm has better convergence behavior than synchronous GMP.

\subsection{Blockwise RGMP}
In the B-RGMP algorithm, each iteration is divided into $M$ timeslots. A variable node randomly selects a timeslot for message updating in each iteration and the selection in different iterations are independent. The B-RGMP algorithm is given in Algorithm 3. 
\par 
For a fixed $M$, the computation time per iteration of the B-RGMP algorithm does not scale with the network size. The reason is that at each timeslot, the B-RGMP algorithm can be implemented in parallel by assigning the message updating of different variable nodes to different processors. Recall that the computational complexity per variable node in each timeslot remains constant when the network size increases. Hence, with a constant number of timeslots, the average computation time of B-RGMP remains constant when the network size increases.

\begin{algorithm}[h]
\caption{Blockwise Randomized Gaussian Message-Passing (B-RGMP) Algorithm}
\label{alg::rgmp}
{
\begin{algorithmic}[1]
\REQUIRE	
$\widehat{\mathbf{H}}$, $\mathbf{y}$
\ENSURE
$\widehat{x}_k$ for all $k$
\STATE Initialize $t = 0, m_{x_k\rightarrow y_n}^{(0)}=0$, $v_{x_k\rightarrow y_n}^{(0)}=1,$ for all $k,n$.
\STATE \textbf{Repeat}
\STATE Set $t \Leftarrow t+1$.
\STATE Draw $K$ times from $\{1,\cdots, M\}$ with replacement, and define the $K$ numbers as $\sigma_t(1),\cdots, \sigma_t(K)$.
\STATE For $m=1,\cdots,M$
\STATE For $k=1,\cdots, K$, and $\sigma_t(k)=m$, $\widehat{H}_{n,{k}}\neq 0$, compute
\begin{flalign}
\small
&v_{y_n\rightarrow x_{k}}^{(t)}
=&&\frac{1}{P{|\widehat{H}_{n,k}|^2}}\bigg(\widehat{N}_0+P\sum_{j:\sigma_t(j)< m} |\widehat{H}_{n,j}|^2 v_{x_{j}\rightarrow y_n}^{(t)}\nonumber
\\
& &&+P\sum_{j\neq k:\sigma_t(j)\geq m} |\widehat{H}_{n,j}|^2v_{x_{j}\rightarrow y_n}^{(t-1)}\bigg)
&\label{eqn:rv1}
\end{flalign}

\begin{flalign}
\small
&m_{y_n\rightarrow x_{k}}^{(t)}
=&&\frac{1}{P^{\frac{1}{2}}{\widehat{H}_{n,{k}}}}\bigg(y_n-P^{\frac{1}{2}}\sum_{j:\sigma_t(j)< m} \widehat{H}_{n,j}m_{x_{j}\rightarrow y_n}^{(t)}\nonumber
\\
& &&-P^{\frac{1}{2}}\sum_{j\neq k:\sigma_t(j)\geq m} \widehat{H}_{n,j}m_{x_{j}\rightarrow y_n}^{(t-1)}\bigg)
&\label{eqn:rm1}
\end{flalign}

\begin{flalign}
\small
&v_{x_{k}\rightarrow y_n}^{(t)}=\bigg( \sum_{\widehat{H}_{j,{k}}\neq 0, j\neq n} \frac{1}{v_{y_j\rightarrow x_{k}}^{(t)}}+1\bigg)^{-1}
&\label{eqn:rv2}
\end{flalign}

\begin{flalign}
\small
&m_{x_{k}\rightarrow y_n}^{(t)}
=v_{x_{k}\rightarrow y_n}^{(t)} \sum_{\widehat{H}_{j,{k}}\neq 0, j\neq n} \frac{{m_{y_j\rightarrow x_{k}}^{(t)}}}{{v_{y_j\rightarrow x_{k}}^{(t)}}}
&\label{eqn:rm2}
\end{flalign}

\STATE \textbf{Until} {\text{stopping criteria is satisfied}}
\STATE Compute 
\begin{flalign}
\small
&v_k=\bigg( \sum_{\widehat{H}_{n,k}\neq 0} \frac{1}{v_{y_n\rightarrow x_k}^{(t)}}+1\bigg)^{-1}
&
\end{flalign}

\begin{flalign}
\small
&\widehat{x}_k=v_k \sum_{\widehat{H}_{n,k}\neq 0} \frac{{m_{y_n\rightarrow x_k}^{(t)}}}{{v_{y_n\rightarrow x_k}^{(t)}}}.
&
\end{flalign}

\end{algorithmic}}
\end{algorithm}
\subsection{Convergence Analysis of B-RGMP}
The convergence condition of the RGMP algorithm can be readily extended to that of the B-RGMP algorithm by changing the distribution of update schedule $\sigma$ accordingly. Moreover, the B-RGMP algorithm allows us to derive a simple convergence condition for the special case with $M=2$, where in each iteration, all messages are updated within two timeslots. We show that in the special case, if the GMP algorithm converges, the expected output $\bm \phi^{(t)}$ of B-RGMP always converges, as formally stated below.
\begin{corollary}
When the number of timeslots $M=2$, the expected output $\bm \phi^{(t)}$ of the B-RGMP algorithm converges to the unique point if and only if the spectral radius $\rho(\frac{3}{4}\mathbf{\Omega}+\frac{1}{4}\mathbf{\Omega}^2)<1$.
\end{corollary}
\begin{proof}
From Theorem 2, it suffices to show that $\mathbf{\Lambda}=\frac{3}{4}\mathbf{\Omega}+\frac{1}{4}\mathbf{\Omega}^2$. Based on (\ref{eqn:R1}), when there are only two timeslots, we obtain $(\bm \Omega-\mathbf{R}_{\sigma})^2=\mathbf{0}$. Then,
\begin{equation}
\begin{aligned}
&\mathbf L_{\sigma}^{-1}\mathbf R_{\sigma}\\
=&(\mathbf{I}-(\bm \Omega-\mathbf{R}_{\sigma}))^{-1}\mathbf{R}_{\sigma}\\
=&(\mathbf{I}+\bm \Omega-\mathbf{R}_{\sigma})\mathbf{R}_{\sigma}\\
=&\mathbf{R}_{\sigma}+(\bm \Omega-\mathbf{R}_{\sigma})(\mathbf{R}_{\sigma}-\bm \Omega+\bm \Omega)\\
=&\mathbf{R}_{\sigma}-\mathbf{R}_{\sigma}\bm \Omega+\bm \Omega^2.
\end{aligned}
\end{equation}
Since the probability of $\sigma_t(k)\leq \sigma_t(j)$ for arbitrary $k$ and $j$ is
\begin{equation}
\begin{aligned}
&P(\sigma_t(k)\leq \sigma_t(j))
\\
=&1-P(\sigma_t(k)>\sigma_t(j))
\\
=&1-\frac{2^{K-2}}{2^K}
\\
=&\frac{3}{4},
\end{aligned}
\end{equation}
the expectation of $\mathbf{R}_{\sigma}$ is $\mathrm{E}_{\sigma}[\mathbf R_{\sigma}]=\frac{3}{4}\bm \Omega$. Hence, 
\begin{equation}
\bm \Lambda=\frac{3}{4}\Omega+\frac{1}{4}\Omega^2.
\end{equation}
%
%
\end{proof}
The B-RGMP algorithm with $M=2$ is not trivial since the convergence is greatly improved compared with GMP. Denote by $\lambda_1,\cdots,\lambda_{NK}$ the eigenvalues of $\mathbf{\Omega}$. Then, the eigenvalues of $\frac{3}{4}\mathbf{\Omega}+\frac{1}{4}\mathbf{\Omega}^2$ are $\frac{3}{4}\lambda_1+\frac{1}{4}\lambda_1^2,\cdots,\frac{3}{4}\lambda_{NK}+\frac{1}{4}\lambda_{NK}^2$. The spectral radius of $\mathbf{\Omega}$ and $\frac{3}{4}\mathbf{\Omega}+\frac{1}{4}\mathbf{\Omega}^2$ are $\max_i |\lambda_i|$ and $\max_i |\frac{3}{4}\lambda_i+\frac{1}{4}\lambda_i^2|$ respectively. Thus, $\rho(\frac{3}{4}\mathbf{\Omega}+\frac{1}{4}\mathbf{\Omega}^2)<1$ always holds if $\rho(\mathbf{\Omega})<1$. But the converse does not hold in general. Take the channel matrix $\mathbf H$ in (\ref{eqn:H}) as an example. The spectral radius of the corresponding $\mathbf \Omega$ is $\rho(\mathbf \Omega)=1.0287$, while $\rho(\frac{3}{4}\mathbf{\Omega}+\frac{1}{4}\mathbf{\Omega}^2)=0.9203$. This means that the expected output of B-RGMP converges while GMP diverges. Therefore, the condition for the expected convergence of the B-RGMP algorithm is less stringent than that of the synchronous message passing.

\section{Numerical Comparisons}
In this section, we compare the performance of RGMP and B-RGMP with other existing algorithms. Unless specified otherwise, we assume that both users and RRHs are uniformly at random located in a circular area with user density $\beta_K=8/\text{km}^2$ and RRH density $\beta_N=10/\text{km}^2$. The path loss exponent is $3.7$, and the average transmit SNR at the user side equals to $95$dB. That is $\frac{P}{N_0}=95$dB. Moreover, the stopping criteria is $\delta^{(t)}<\delta$, where $\delta^{(t)}$ is the relative error after the $t$-th iteration, where one iteration means all the messages are updated once. In particular, $\delta^{(t)}=\frac{\|P\mathbf{H}^H\mathbf{H}\mathbf{x}^{(t)}-{P}^{\frac{1}{2}}\mathbf{H}^H\mathbf{y}\|}{\|{P}^{\frac{1}{2}}\mathbf{H}^H\mathbf{y}\|}$, with $\mathbf{x}^{(t)}$ being the estimated transmitted signal after $t$ iteration. 
\subsection{Comparison of Convergence}

\begin{figure}[!ht]
\centering
{\includegraphics[width=0.48\textwidth]{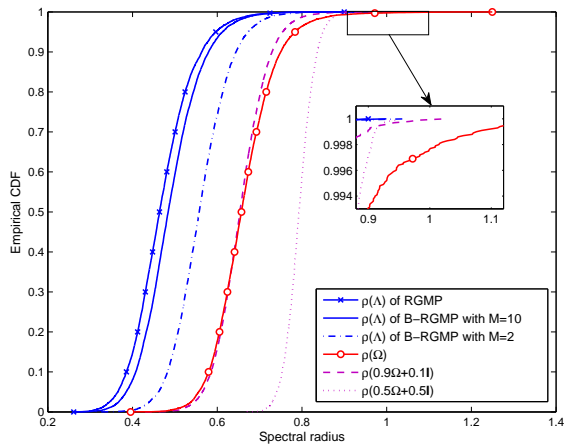}}
\caption{Cumulative distribution function of the spectral radius with $N=20$, $K=15$, and $\frac{P}{N_0}=95$dB.}\label{fig:cdf_rho}
\end{figure}  
In this subsection, we compare the convergence of RGMP with conventional GMP. From Lemma 1 and Theorem 1, both GMP and RGMP have guaranteed convergence for variances. However, they have different conditions to ensure the convergence of means: GMP requires $\rho(\mathbf{\Omega})<1$ while RGMP requires $\rho(\mathbf{\Lambda})<1$. Since both $\mathbf{\Omega}$ and $\mathbf{\Lambda}$ highly depend on the network geometry, it is difficult to theoretically compare these two convergence conditions. To shed light on the difference, we plot the cumulative distribution function (CDF) of the spectral radius of $\mathbf{\Omega}$ and $\mathbf{\Lambda}$ in Fig.~{\ref{fig:cdf_rho}}. We assume that users and RRHs are randomly located in a circular network area with radius $r$. The user density is $\beta_K=8/\text{km}^2$ and the RRH density is $\beta_N=10/\text{km}^2$. We see that $\rho(\mathbf{\Lambda})$ is more likely to take small values than $\rho(\mathbf{\Omega})$. This implies that RGMP converges with a higher probability than GMP. Indeed, we have run over $10000$ times for each setting, and $\rho(\mathbf{\Lambda})<1$ for all the cases. We also plot the CDF of $\rho(\bm \Lambda)$ for B-RGMP with $M=2$ and $M=10$. We see that when the number of timeslots $M$ is moderately large (e.g., $M$=10), the distributions of $\rho(\bm \Lambda)$ for RGMP and B-RGMP are quite close to each other. This implies that the B-RGMP algorithm is a reasonable generalization of RGMP without degrading the convergence performance.
\par 
In Fig.~\ref{fig:cdf_rho}, we also compare the convergence of RGMP with damping-based GMP \cite{su2015convergence,liu2016convergence}. In the GMP algorithm with damping, a message from the check nodes is a weighted average between the old message and the new message. That is, the eqn (\ref{eqn:m1}) of Algorithm 1 is replaced by the following equation
\begin{equation}
\begin{aligned}
m_{y_n\rightarrow x_k}^{(t)}
=\eta\frac{y_n-P^{\frac{1}{2}}\sum_{j\neq k} \widehat{{H}}_{n,j}m_{x_j\rightarrow y_n}^{(t-1)}}{P^{\frac{1}{2}}{\widehat{{H}}_{n,k}}}+(1-\eta) m_{y_n\rightarrow x_k}^{(t-1)},
\end{aligned}
\end{equation}
where $\eta$ is the damping factor. Consequently, the convergence condition of GMP with damping now becomes $\rho(\eta\mathbf{\Omega}+(1-\eta)\mathbf{I})<1$. As we can see from Fig.~{\ref{fig:cdf_rho}}, when the damping factor is large (i.e., $\eta=0.9$ in Fig.~{\ref{fig:cdf_rho}}), the probability that the spectral radius exceeds $1$ is non-zero. When the damping factor is small, the spectral radius is less likely to exceed $1$. However, as shown in our later simulations, the convergence rate decreases when the damping factor descreases. There exists a trade-off between the convergence probability and the convergence rate of GMP with damping. How to efficiently determine the value of the damping factor is still an open problem.

\subsection{Comparison of Convergence Speed}
In this subsection, we compare the convergence speed of RGMP and B-RGMP with other algorithms including ADMM \cite{sun2015on}, GAMP \cite{rangan2011generalized}, GMP with damping \cite{su2015convergence}, and conjugate gradient (CG) \cite{barrett1994templates}. For a fair comparison, the channel sparsification approach with distance threshold $d_0=1000$m is adopted in all algorithms. In this way, all the algorithms have a linear per-iteration computational complexity with the network size. Thus, we only focus on the convergence speed of these algorithms.
\begin{figure}[!h]
\centering
{\includegraphics[width=0.48\textwidth]{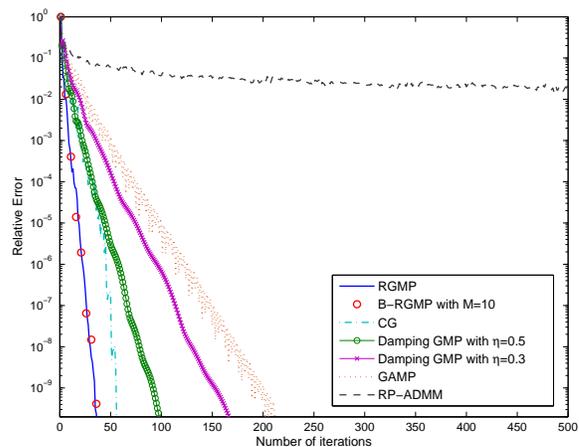}}
\caption{Relative error vs number of iterations when the number of RRHs $N=40$.}\label{fig:err_admm}
\end{figure}
\par 

In Fig.~\ref{fig:err_admm}, the relative error $\delta^{(t)}$ is plotted against the number of iterations for $N=40$ and $K=32$. We see that RGMP, B-RGMP, GMP with damping, and CG converge relatively fast. For example,  the relative error of RGMP reduces to $0.001$ within $10$ iterations. However, the performance of the ADMM algorithm is unsatisfactory. Over $500$ iterations are needed for the ADMM algorithm to reduce the relative error to $0.02$. In fact, from simulation results not presented here, ADMM requires over $5000$ iterations on average to reduce the relative error to $0.001$ for the network configuration in Fig.~\ref{fig:err_admm}. Therefore, even though the computational complexity per iteration of ADMM is linear in the number of RRHs and the expected convergence is guaranteed \cite{sun2015on}, it is still impractical to adopt the ADMM algorithm in C-RAN due to the extremely slow convergence.
 \begin{figure}[!ht]
\centering
{\includegraphics[width=0.48\textwidth]{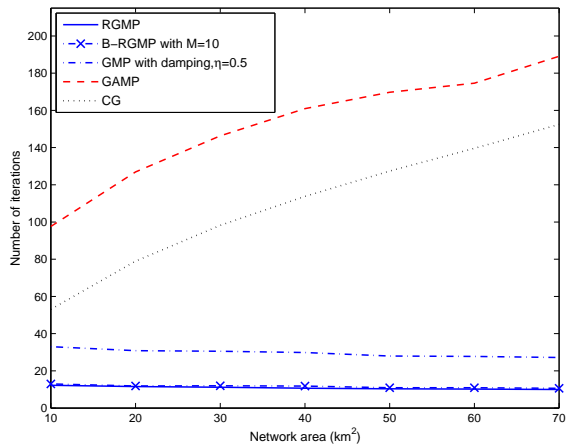}}
\caption{Convergence speed against the network size.}\label{fig:rate_size}
\end{figure}
\par 
In Fig.~\ref{fig:rate_size}, we plot the convergence speed of both RGMP and B-RGMP against the network size, where the convergence speed is measured by the critical number of iterations to achieve $\delta^{(t)}<10^{-5}$. Due to the extremely slow convergence speed of ADMM as shown in Fig.~\ref{fig:err_admm}, we ignore ADMM and only plot the convergence speed of GMP with damping (with $\delta^{(t)}<10^{-5}$) and GAMP (with $\delta^{(t)}<10^{-3}$) for comparison. We observe that the number of iterations needed by GAMP grows roughly linearly with the network size. In contrast, the convergence speeds of both RGMP, B-RGMP, and GMP with damping are constant with the network size. Note that the computational complexity per iteration of GAMP/GMP with damping/RGMP/B-RGMP is linear in the network size. Thus, the total computational complexity of both GMP with damping and the RGMP/B-RGMP algorithm is linear in the network size, while that of GAMP grows quadratically with the network size. Moreover, with parallel implementation, the computation time of B-RGMP remains constant with the network size.
\par 
We emphasize that even though its performance looks not bad in simulation, GMP with damping has several drawbacks compared with the RGMP algorithm. For example, how to efficiently determine the value of the damping factor is still an open problem. In Fig.~\ref{fig:err_admm}, we observe that the GMP with damping converges faster when the damping factor increases. Recall that the spectral radius is more likely to exceed $1$ when the damping factor increases. Indeed, there exists a trade-off between the convergence probability and the convergence speed of GMP with damping. In previous works, the damping factor is usually determined through simulations \cite{moretti2014on}. Considering the large network size of C-RAN, empirically calculating the damping factor introduces unaffordable complexity cost. A recent work \cite{su2015convergence} derived a range of the damping factors, which guarantees the convergence of GMP. The range, however, is a function of the eigenvalues of $\mathbf \Omega$, which means choosing the damping factor based on \cite{su2015convergence} still requires prohibitively high computational complexity.

\subsection{Comparison of Performance}
\begin{figure}[!h]
\centering
{\includegraphics[width=0.48\textwidth]{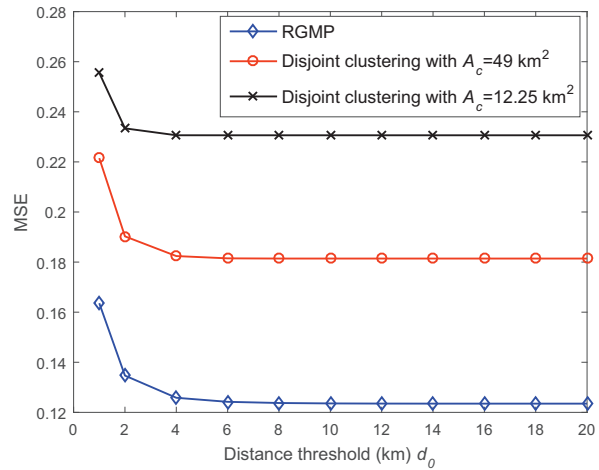}}
\caption{SINR ratio vs the distance threshold $d_0$ when the network area is 200km$^2$.}\label{fig:disjoint}
\end{figure}
In this subsection, we compare the performance of the RGMP algorithm with a disjoint clustering algorithm. The disjoint clustering algorithm divides the whole network into disjoint square clusters with area $A_c$, and do MMSE detection independently in each disjoint cluster. Channel sparsification is also applied in the disjoint clustering algorithm. In Fig.~\ref{fig:disjoint}, we plot the mean squared error (MSE) against the distance threshold, where MSE refers to $\mathrm{E}[|x_k-\hat x_k|^2]$. The network area is 200km$^2$. The numbers of RRHs and users are $2000$ and $1600$, respectively. We see that the gap between the RGMP algorithm and the disjoint clustering algorithm is very large. For example, when the distance threshold $d_0$ is $4$km, the MSE of RGMP is less than $0.13$, which is only half of the MSE of the disjoint clustering algorithm with cluster area $49$km$^2$.

\section{Conclusions}
In this paper, we proposed RGMP and B-RGMP for scalable uplink signal detection in C-RANs. With channel sparsification, signal detection in a C-RAN was converted to an inference problem over a bipartite random geometric graph. A random message-update schedule was employed to address the convergence issue of GMP over a bipartite random geometric graph. We analysed the convergence condition of the proposed RGMP algorithm and showed that the convergence condition of RGMP is much less stringent than that of GMP. Numerical results demonstrated that RGMP exhibits much faster convergence than the existing algorithms, such as GAMP and ADMM. We further proposed the B-RGMP algorithm for parallel implementation. With a fixed number of timeslots, the total computation time of B-RGMP does not increase with the network size, which means B-RGMP is a perfectly scalable detection algorithm. The work in this paper sheds light on the design of message-passing algorithms on general loopy graphs, which has been a challenging topic in the field for years. Future work can be done in a number of interesting directions. For example, message passing has been applied to reduce the complexity of signal detection with constellation constraints \cite{wu2014low,som2010improved}. The convergence of these algorithms can be potentially improved by introducing randomized message updating. Moreover, RGMP can be extended to the design of uplink signal detectors with limited fronthaul capacity, as well as to the design of downlink beamforming for C-RANs. These topics are worthy of our future research endeavour.

\end{document}